\documentclass{article}

\usepackage{arxiv}
\usepackage[utf8]{inputenc}
\usepackage[T1]{fontenc}
\usepackage{amsmath,amssymb,amsthm,mathtools}
\usepackage{bbm}
\usepackage{booktabs}
\usepackage{enumitem}
\usepackage{microtype}
\usepackage[colorlinks=true,linkcolor=blue,citecolor=blue,urlcolor=blue]{hyperref}

\title{Nonlinear Two-Time-Scale Stochastic Approximation: A Sharp Phase Transition and How to Beat It}

\author{Dhruv Sarkar$^{1,2}$  \quad Vaneet Aggarwal$^3$\\
$^1$Indian Institute of Technology Kharagpur\\
$^2$Mohamed bin Zayed University of Artificial Intelligence\\
$^3$Purdue University\\
\texttt{dhruv.sarkar@gmail.com}\quad  \texttt{vaneet@purdue.edu}}

\newtheorem{theorem}{Theorem}[section]
\newtheorem{lemma}[theorem]{Lemma}
\newtheorem{corollary}[theorem]{Corollary}
\newtheorem{proposition}[theorem]{Proposition}
\newtheorem{assumption}[theorem]{Assumption}
\newtheorem{remark}[theorem]{Remark}

\newcommand{\R}{\mathbb{R}}
\newcommand{\E}{\mathbb{E}}

\newcommand{\norm}[1]{\left\lVert #1\right\rVert}
\newcommand{\ip}[2]{\left\langle #1,#2\right\rangle}

\begin{document}
\maketitle

\begin{abstract}
Recent finite-time analyses of nonlinear two-time-scale stochastic approximation show that under contractive assumptions the slow iterate $Y_k$ with stepsizes $\beta_k=\Theta(k^{-1})$ and $\alpha_k=\Theta(k^{-a})$, $a\in(1/2,1)$, generally satisfies a mean-square rate of order $k^{-a}$; decoupled $k^{-1}$ rates require strong local linearity.  
We identify a sharp regularity-dependent boundary. In a rate-determining normal form where the slow drift contains a locally linear leakage and a nonlinear remainder of order $1+\rho$ ($\rho\in[0,1]$), the uncorrected recursion satisfies
\[
\mathbb{E}\|Y_k\|^2 \le C\bigl(k^{-1}+k^{-a(1+\rho)}\bigr),
\]
and a matching scalar Gaussian lower bound shows that the slower term is unavoidable without modifying the update. Thus the decoupled $k^{-1}$ rate is guaranteed for the uncorrected recursion exactly when $a(1+\rho)\ge 1$.  
This lower bound concerns only the naive update; it is not an information-theoretic obstruction. We demonstrate this by equipping the normal-form recursion with an auxiliary online bias estimator
\[
M_{k+1}=M_k+\gamma_k(R(X_k)-M_k),\qquad \beta_k\ll\gamma_k\ll\alpha_k,
\]
and subtracting $M_k$ from the slow update. Under the same stability, moment, and remainder assumptions, the corrected recursion achieves $\mathbb{E}\|\widetilde Y_k\|^2=O(k^{-1})$ for every $\rho\in[0,1]$, including regimes where the uncorrected update provably suffers the slower rate.  
Finally, we prove localized transfer theorems that extend the phase-transition mechanism to general nonlinear TTSA in fast-manifold coordinates. The proofs are non-asymptotic and rely on two Abel-transform cancellations: one for the locally linear fast-error leakage, and one for the tracked nonlinear bias.
\end{abstract}


\section{Introduction}

Two-time-scale stochastic approximation (TTSA) studies coupled recursions in which a fast iterate is updated with stepsize $\alpha_k$ and a slow iterate is updated with a smaller stepsize $\beta_k$. In its fixed-point form, the recursion is often written as
\begin{align*}
  x_{k+1} &= x_k + \alpha_k\big(f(x_k,y_k)-x_k+M_{k+1}\big),\\
  y_{k+1} &= y_k + \beta_k\big(g(x_k,y_k)-y_k+M'_{k+1}\big),
\end{align*}
where $\beta_k/\alpha_k\to0$ in the truly two-time-scale regime. TTSA is a central tool in reinforcement learning, stochastic optimization, stochastic control, saddle-point learning, and game control; see, for example, the classical monographs and early analyses \cite{robbins1951,kushner2003,borkar2008,konda2004}.

For linear TTSA, the slow iterate may enjoy a decoupled $O(k^{-1})$ mean-square rate under $\beta_k=\Theta(k^{-1})$ even when $\alpha_k=\Theta(k^{-a})$ with $a<1$ \cite{konda2004,kaledin2020,haque2024}. For nonlinear contractive TTSA the situation is more delicate. Doan \cite{doan2023} established finite-time guarantees for nonlinear TTSA, and a recent averaged-noise analysis of Chandak \cite{chandak20261} improved the truly separated nonlinear rate to $O(k^{-a})$ under standard contractive assumptions.  Meanwhile, Han, Li, and Zhang \cite{han2024} proved that the decoupled $O(k^{-1})$ rate is recovered under \emph{nested local linearity} assumptions, and provided empirical evidence that local linearity may be necessary. Chandak \cite{chandak20261} explicitly posed it as an open problem whether $O(k^{-1})$ is possible for nonlinear TTSA \emph{without additional assumptions}, conjecturing that $O(k^{-a})$ is the best achievable rate in general.

%Chandak \cite{chandak20261} argued that their results were possibly optimal because Han, Li, and Zhang \cite{han2024} showed that decoupled rates can be recovered under nested local linearity assumptions and supplied empirical evidence that local linearity may be necessary for decoupling. However, a formal theoretical guarantee for the same is missing. 

This leaves a structural question:
\begin{quote}
\emph{What feature of the slow drift determines whether the slow iterate converges at the decoupled $k^{-1}$ rate or at the fast-tracking rate $k^{-a}$?}
\end{quote}
It also raises an algorithmic open problem: if the slow rate is limited by nonlinear bias, can a change in the update rule remove that bias and recover the decoupled rate without adding simulator access or oracle debiasing?

We answer both questions using a \textbf{local normal-form} model that isolates the rate-determining terms from the otherwise involved mechanics of the general TTSA recursion. Let $X_k$ denote the fast tracking error and $Y_k$ the slow error. After linearizing the fast dynamics and expanding the slow drift in the fast error, the leading model is
\begin{align}
  X_{k+1} &= (I-\alpha_k A)X_k+\alpha_k\xi_{k+1}, \label{eq:intro-fast}\\
  Y_{k+1} &= (I-\beta_k B)Y_k+
    \beta_k\big(HX_k+R(X_k)+\zeta_{k+1}\big). \label{eq:intro-slow}
\end{align}
Here $HX_k$ is the locally linear leakage of fast error into the slow recursion, and $R(X_k)$ is the genuinely nonlinear remainder. We assume
\[
  \norm{R(x)} \le L_R\norm{x}^{1+\rho},\qquad \rho\in[0,1],
\]
where $\rho$ measures local regularity: $\rho=0$ is a Lipschitz remainder, $\rho=1$ a second‑order remainder.

%The parameter $\rho$ measures the local regularity of the nonlinear remainder: $\rho=0$ corresponds to a merely Lipschitz remainder, while $\rho=1$ corresponds to a second-order remainder.

Prior work \cite{chandak20261,han2024} suggested that decoupled $O(k^{-1})$ slow convergence requires local linearity, and that without local linearity the best possible rate is $O(k^{-a})$. Our results show that this binary picture is incomplete. The determining quantity is the local regularity and predictable bias of the nonlinear remainder. For an order-$(1+\rho)$ remainder, the sharp uncorrected rate is $O(k^{-1} + k^{-a(1+\rho)})$, with a matching lower bound. Thus local linearity is sufficient but not necessary: sufficiently regular nonlinear remainders still recover the decoupled rate.

\paragraph{Sharp rate for the uncorrected recursion.}
For $\alpha_k=\alpha_0(k+k_0)^{-a}$ and $\beta_k=\beta_0(k+k_0)^{-1}$, $a\in(1/2,1)$, our first main theorem proves
\begin{equation}
  \E\norm{Y_k}^2 \le C\left((k+k_0)^{-1}+(k+k_0)^{-a(1+\rho)}\right). \label{eq:intro-rate}
\end{equation}
The term $(k+k_0)^{-1}$ is the usual stochastic approximation noise floor for the slow recursion. The term $(k+k_0)^{-a(1+\rho)}$ is the squared nonlinear bias induced by the fast tracking fluctuations. Thus the uncorrected slow recursion exhibits the phase transition
\[
  \E\norm{Y_k}^2 =
  \begin{cases}
    O(k^{-a(1+\rho)}), & a(1+\rho)<1,\\
    O(k^{-1}), & a(1+\rho)\ge 1.
  \end{cases}
\]

The following table summarizes the regimes.
\begin{table}[h]
\centering
\begin{tabular}{lll}
\toprule
Slow-drift structure in fast error & Regularity & Uniform rate for uncorrected recursion \\
\midrule
Exactly locally linear, $R\equiv0$ & no nonlinear bias & $O(k^{-1})$ \\
Merely Lipschitz remainder & $\rho=0$ & $O(k^{-a})$ \\
$C^{1,\rho}$-type remainder & $\rho\in(0,1)$ & $O(k^{-\min\{1,a(1+\rho)\}})$ \\
Second-order Taylor remainder & $\rho=1$ & $O(k^{-1})$ for all $a>1/2$ \\
\bottomrule
\end{tabular}
\caption{Regularity regimes for the uncorrected normal-form recursion. The decoupled slow rate is guaranteed uniformly exactly when $a(1+\rho)\ge1$.}
\label{tab:regularity-regimes}
\end{table}

\paragraph{Why the linear leakage does not hurt.}
The fast variable has $\E\norm{X_k}^2=\Theta(\alpha_k)=\Theta(k^{-a})$. A direct input-to-state bound on the term $HX_k$ would therefore lose the desired $k^{-1}$ rate. The key observation, also underlying averaged-noise analyses of nonlinear TTSA, is that $X_k$ itself satisfies a stable stochastic approximation recursion. Consequently,
\[
  AX_k = \alpha_k^{-1}(X_k-X_{k+1})+\xi_{k+1}.
\]
After substituting this identity into the slow convolution, the contribution of $HX_k$ splits into a martingale term with variance $O(k^{-1})$ and a telescoping Abel-transform term that is also $O(k^{-1})$. This is the precise mechanism by which local linearity restores decoupling.

\paragraph{Why the nonlinear remainder is rate-limiting.}
In contrast, if $R(x)$ is of order $\norm{x}^{1+\rho}$, then
\[
  \E\norm{R(X_k)} \asymp \alpha_k^{(1+\rho)/2}\asymp k^{-a(1+\rho)/2}
\]
in Gaussian normal-form examples. This contribution is not a martingale fluctuation and therefore does not average down to $k^{-1}$; it appears as a slowly varying bias in the slow recursion. We prove that this obstruction is not an artifact of the proof by constructing a one-dimensional Gaussian instance for which
\begin{equation}
  \E[Y_k^2] \ge c_1 k^{-a(1+\rho)}+c_2 k^{-1}. \label{eq:intro-lower}
\end{equation}
The lower bound uses $H=0$, so it is caused solely by nonlinear curvature, not by linear leakage. It is a lower bound for the uncorrected update rule \eqref{eq:intro-slow}; it does not preclude algorithms that estimate and subtract the predictable nonlinear bias.

\paragraph{Online bias tracking resolves the algorithmic question in the normal form.}
The lower bound identifies what must be removed. We introduce an auxiliary bias tracker
\begin{equation}
  M_{k+1}=M_k+\gamma_k\big(R(X_k)-M_k\big),
  \qquad
  \gamma_k=\gamma_0(k+k_0)^{-c}, \label{eq:intro-bias-tracker}
\end{equation}
with an intermediate time scale $\beta_k\ll\gamma_k\ll\alpha_k$. The corrected slow update is
\begin{equation}
  \widetilde Y_{k+1}
  =(I-\beta_kB)\widetilde Y_k
  +\beta_k\big(HX_k+R(X_k)-M_k+\zeta_{k+1}\big). \label{eq:intro-debiased}
\end{equation}
This is a fully online correction for the normal-form recursion: it uses only the observed nonlinear term $R(X_k)$, adds no inner fast loops, and assumes no simulator or oracle access. The theorem proves
\[
  \E\norm{\widetilde Y_k}^2\le Ck^{-1}
\]
for every $\rho\in[0,1]$, including the regimes $a(1+\rho)<1$ where the uncorrected recursion has the lower bound \eqref{eq:intro-lower}. The proof uses a second Abel transform. Since
\[
  R(X_k)-M_k=\gamma_k^{-1}(M_{k+1}-M_k),
\]
the slow convolution of the corrected nonlinear input telescopes with coefficient $\beta_k/\gamma_k$; the intermediate time scale is chosen so that the resulting boundary and variation terms are below the $k^{-1}$ noise floor.

\paragraph{Technical novelty.}
The paper has two related technical messages. First, the same Abel-transform idea that averages locally linear fast-error leakage also gives a sharp diagnostic for nonlinear curvature: only the predictable part of the nonlinear input survives slow averaging. Second, once this predictable part is explicitly tracked on an intermediate time scale, it can be removed by another Abel-transform cancellation. This proves that the lower bound is an obstruction for the raw recursion, not an information-theoretic limitation of the normal-form problem class.

\paragraph{Contributions.}
The paper makes five technical contributions.
\begin{enumerate}[leftmargin=*]
\item We formulate a normal-form TTSA model that separates the locally linear fast-error leakage from the genuinely nonlinear slow-drift remainder.
\item We prove a finite-time upper bound with the regularity-dependent rate $O(k^{-1}+k^{-a(1+\rho)})$. The proof is non-asymptotic and relies on an Abel-transform form of averaged-noise cancellation.
\item We prove a matching lower bound in a scalar Gaussian subclass. Hence, for the uncorrected normal-form recursion, the exponent $\min\{1,a(1+\rho)\}$ cannot be improved uniformly.
\item We introduce an online three-time-scale bias-tracking recursion that subtracts a running estimate of the nonlinear predictable bias. Under the same normal-form stability, moment, and regularity assumptions, the corrected slow iterate satisfies $\E\norm{\widetilde Y_k}^2=O(k^{-1})$ for all $\rho\in[0,1]$.
\item We transfer the normal-form phase transition to nonlinear TTSA and clarify its mechanism. A centering/debiasing result shows that nonlinear fluctuations of the same instantaneous size recover the decoupled $k^{-1}$ rate when they are martingale-centered. Localized transfer and verification theorems in fast-manifold coordinates then show that the same phase transition persists under moving fast equilibria, induced slow-noise perturbations, variable local linearizations, and both separable and mixed higher-order Taylor terms, whenever the required finite-time localized moment estimates hold.
\end{enumerate}

\section{Local normal form}

We work on a filtered probability space $(\Omega,\mathcal F,\{\mathcal F_k\}_{k\ge0},\mathbb P)$. All random variables below are assumed to be adapted to this filtration. For a positive integer $d$, $\norm{\cdot}$ denotes the Euclidean norm on $\R^d$ and the induced operator norm on matrices.

Fix $a\in(1/2,1)$ and define
\begin{equation}
  \alpha_k = \frac{\alpha_0}{(k+k_0)^a},\qquad
  \beta_k = \frac{\beta_0}{k+k_0},\qquad k\ge0, \label{eq:stepsizes}
\end{equation}
where $\alpha_0,\beta_0>0$. The offset $k_0\ge2$ is chosen once and for all large enough so that all contraction inequalities below hold and so that every small-step comparison used in the proofs is valid. In particular, we require $\lambda\alpha_k\le1/4$, $b\beta_k\le1/2$, $2a/\bar k\le(\lambda/2)\alpha_k$, and the fourth-moment comparison \eqref{eq:alpha2-compare}. Such a finite choice exists because $a<1$ and $\bar k^{-1}=o(\alpha_k)$. If a smaller offset is used, the same rates hold after a finite burn-in. We write
\[
  \bar k \coloneqq k+k_0.
\]

The normal form is the coupled recursion
\begin{align}
  X_{k+1} &= (I-\alpha_k A)X_k+\alpha_k\xi_{k+1}, \label{eq:fast}\\
  Y_{k+1} &= (I-\beta_k B)Y_k+
             \beta_k\big(HX_k+R(X_k)+\zeta_{k+1}\big), \label{eq:slow}
\end{align}
where $X_k\in\R^{d_x}$ is the fast tracking error and $Y_k\in\R^{d_y}$ is the slow error. The matrices $A\in\R^{d_x\times d_x}$ and $B\in\R^{d_y\times d_y}$ describe the local stable linearizations of the fast and slow dynamics. The matrix $H\in\R^{d_y\times d_x}$ is the first-order dependence of the slow drift on the fast tracking error. The map $R:\R^{d_x}\to\R^{d_y}$ is the nonlinear remainder.

\begin{assumption}[Stability and noise]\label{ass:stability}
There exist constants $\lambda,b>0$ such that, for all $k\ge0$,
\begin{equation}
  \norm{I-\alpha_k A}\le 1-\lambda\alpha_k,
  \qquad
  \norm{I-\beta_k B}\le 1-b\beta_k. \label{eq:matrix-contraction}
\end{equation}
The matrix $A$ is invertible. The noise sequences are martingale differences:
\begin{equation}
  \E[\xi_{k+1}\mid\mathcal F_k]=0,
  \qquad
  \E[\zeta_{k+1}\mid\mathcal F_k]=0. \label{eq:mds}
\end{equation}
Moreover, for a finite constant $\sigma_4$,
\begin{equation}
  \E[\norm{\xi_{k+1}}^4\mid\mathcal F_k]\le \sigma_4^4,
  \qquad
  \E[\norm{\zeta_{k+1}}^2\mid\mathcal F_k]\le \sigma_4^2. \label{eq:moment-noise}
\end{equation}
The initial conditions $X_0,Y_0$ are deterministic.
\end{assumption}

\begin{assumption}[Remainder regularity]\label{ass:remainder}
For some $\rho\in[0,1]$ and $L_R<\infty$,
\begin{equation}
  \norm{R(x)}\le L_R\norm{x}^{1+\rho},\qquad x\in\R^{d_x}. \label{eq:remainder}
\end{equation}
\end{assumption}

The quantity
\begin{equation}
  r\coloneqq b\beta_0,
  \qquad
  s\coloneqq a(1+\rho) \label{eq:rs-def}
\end{equation}
will appear throughout the paper. The exponent $r$ is the polynomial strength of the slow linear contraction, while $s$ is the exponent generated by the nonlinear remainder.

\begin{assumption}[Slow contraction strength]\label{ass:strength}
The slow contraction exponent satisfies
\begin{equation}
  r>\frac12. \label{eq:r-condition}
\end{equation}
\end{assumption}

This is the standard finite-time requirement that the linear part of the slow recursion be strong enough to wash out the initial condition and to make the slow stochastic convolution square summable at rate $k^{-1}$. No additional condition involving the nonlinear exponent $s=a(1+\rho)$ is needed. When $s>1$, the nonlinear convolution need only be controlled at the $k^{-1}$ noise floor; the more flexible weighted-convolution estimate in the appendix handles this regime.

\begin{remark}[Relation to nonlinear TTSA]
For a nonlinear TTSA recursion, set $X_k=x_k-x^*(y_k)$ and $Y_k=y_k-y^*$. A local expansion of the slow drift around $(x^*(y),y)$ has the form
\[
  g(x,y)-y^* = g(x^*(y),y)-y^* + H(y)(x-x^*(y)) + R(x-x^*(y),y).
\]
The normal form \eqref{eq:fast}--\eqref{eq:slow} freezes the stable linear parts at the fixed point and isolates the two fast-error mechanisms that matter for rates: the locally linear term and the nonlinear remainder. The normal-form theorem is a sharp analysis of the leading local mechanism. Section~4 gives perturbation conditions under which this mechanism transfers to nonlinear TTSA in fast-manifold coordinates, while global stability or localization is handled separately.
\end{remark}

\section{Main results}

\begin{theorem}[Regularity-dependent upper bound]\label{thm:upper}
Suppose Assumptions \ref{ass:stability}, \ref{ass:remainder}, and \ref{ass:strength} hold. Then there exists a finite constant $C$ depending only on
\[
  A,B,H,L_R,\lambda,b,\alpha_0,\beta_0,k_0,a,\rho,\sigma_4,X_0,Y_0
\]
such that, for every $k\ge1$,
\begin{equation}
  \E\norm{Y_k}^2
  \le C\left(\bar k^{-1}+\bar k^{-s}\right),
  \qquad s=a(1+\rho). \label{eq:upper-bound}
\end{equation}
Moreover,
\begin{equation}
  \E\norm{X_k}^2\le C\alpha_k,
  \qquad
  \E\norm{X_k}^{2+2\rho}\le C\alpha_k^{1+\rho}. \label{eq:fast-moments-main}
\end{equation}
\end{theorem}

The proof is given in Appendix~\ref{app:upper}.

Theorem \ref{thm:upper} separates the two contributions to the slow error. The locally linear leakage $HX_k$ contributes only $O(\bar k^{-1})$. The nonlinear remainder contributes at most $O(\bar k^{-1}+\bar k^{-s})$; it can dominate the slow noise floor exactly when $s<1$.

\begin{corollary}[Phase transition]\label{cor:phase}
Under the assumptions of Theorem \ref{thm:upper},
\begin{equation}
  \E\norm{Y_k}^2 \le
  \begin{cases}
    C\bar k^{-a(1+\rho)}, & a(1+\rho)<1,\\[2mm]
    C\bar k^{-1}, & a(1+\rho)\ge1.
  \end{cases} \label{eq:phase-upper}
\end{equation}
In particular, the slow iterate has the decoupled $O(k^{-1})$ rate whenever
\begin{equation}
  \rho\ge \frac{1}{a}-1. \label{eq:threshold}
\end{equation}
For $\rho=0$, the theorem gives $\E\norm{Y_k}^2=O(k^{-a})$. For $\rho=1$ and $a>1/2$, it gives $\E\norm{Y_k}^2=O(k^{-1})$.
\end{corollary}

\begin{remark}[Uniform interpretation of the threshold]\label{rem:uniform-threshold}
The word ``threshold'' is used in a uniform worst-case sense over the normal-form class. Special instances below the threshold can converge faster, for example when $R\equiv0$ or when the nonlinear remainder has a symmetry that cancels its leading mean. The lower bound below shows that no rate better than $\bar k^{-\min\{1,a(1+\rho)\}}$ can be guaranteed uniformly from Assumptions~\ref{ass:stability}--\ref{ass:strength} alone.
\end{remark}

The next theorem shows that the exponent in Corollary~\ref{cor:phase} is intrinsic to the normal-form class.

\begin{theorem}[Matching scalar lower bound]\label{thm:lower}
Let $d_x=d_y=1$, $H=0$, and consider
\begin{align}
  X_{k+1} &= (1-\lambda\alpha_k)X_k+\alpha_k\xi_{k+1}, \label{eq:lower-fast}\\
  Y_{k+1} &= (1-b\beta_k)Y_k+\beta_k\big(\gamma |X_k|^{1+\rho}+\tau\eta_{k+1}\big), \label{eq:lower-slow}
\end{align}
where $X_0=Y_0=0$, $\lambda,b,\gamma>0$, $\tau\ge0$, and $\{\xi_k\}$ and $\{\eta_k\}$ are mutually independent i.i.d. standard Gaussian sequences. Choose $k_0$ large enough that $0<\lambda\alpha_k<1/4$ and $0<b\beta_k<1/2$ for every $k\ge0$. Then there exist constants $c>0$ and $k_1\ge1$, depending on the fixed parameters but not on $k$, such that for all $k\ge k_1$,
\begin{equation}
  \E[Y_k^2]\ge c\gamma^2\bar k^{-a(1+\rho)}+c\tau^2\bar k^{-1}. \label{eq:lower-bound}
\end{equation}
Consequently, the rate in Theorem \ref{thm:upper} is unimprovable uniformly over the normal-form class whenever $\gamma>0$ and, for the $k^{-1}$ noise floor, $\tau>0$. If in addition $b\beta_0>1/2$, the lower-bound instance belongs to the same slow-contraction regime as Theorem~\ref{thm:upper}.
\end{theorem}

The proof is given in Appendix~\ref{app:lower}.

\begin{remark}[Interpretation]
The lower bound uses $H=0$. Thus the obstruction to $k^{-1}$ for $a(1+\rho)<1$ is not imperfect averaging of the locally linear term. It is the nonzero mean of the nonlinear remainder generated by the fast fluctuation scale $\norm{X_k}\approx\sqrt{\alpha_k}$. The statement is not that every nonlinear remainder creates this obstruction; odd or otherwise cancelling remainders may behave like centered noise. The result is a minimax sharpness statement.
\end{remark}

The next proposition makes the last sentence precise in a form that is useful for bias correction. It separates the size of a nonlinear fluctuation from its predictable component. A nonlinear term of size $\alpha_k^{(1+\rho)/2}$ causes the slower $k^{-a(1+\rho)}$ rate only when it has a persistent predictable part of that size.

\begin{proposition}[Centering and debiasing principle]\label{prop:centering}
Consider the same fast recursion \eqref{eq:fast}. Replace the slow recursion \eqref{eq:slow} by
\begin{equation}
  Y_{k+1}=(I-\beta_kB)Y_k+\beta_k\big(HX_k+m_k+\nu_{k+1}+\zeta_{k+1}\big), \label{eq:centered-slow}
\end{equation}
where $m_k$ is $\mathcal F_k$-measurable and $\nu_{k+1}$ is a martingale difference satisfying
\begin{equation}
  \E[\nu_{k+1}\mid \mathcal F_k]=0,
  \qquad
  \E[\norm{\nu_{k+1}}^2\mid \mathcal F_k]\le K_\nu\alpha_k^{1+\rho}. \label{eq:nu-assumption}
\end{equation}
Assume also that, for some $q>0$ and $K_m<\infty$,
\begin{equation}
  \big(\E\norm{m_k}^2\big)^{1/2}\le K_m\bar k^{-q}. \label{eq:m-assumption}
\end{equation}
Let
\begin{equation}
  \Psi_q(k)=
  \begin{cases}
    \bar k^{-q}, & q<r,\\
    \bar k^{-r}\log(\bar k), & q=r,\\
    \bar k^{-r}, & q>r.
  \end{cases} \label{eq:psi-q}
\end{equation}
Under Assumptions~\ref{ass:stability} and \ref{ass:strength}, there is a finite constant $C$ such that, for all $k\ge1$,
\begin{equation}
  \E\norm{Y_k}^2\le C\left(\bar k^{-1}+\Psi_q(k)^2\right). \label{eq:centering-rate}
\end{equation}
In particular, if $m_k\equiv0$, then $\E\norm{Y_k}^2\le C\bar k^{-1}$.
\end{proposition}

\subsection{Online bias tracking in the normal form}
\label{subsec:online-bias-tracking}

Proposition~\ref{prop:centering} shows that a nonlinear input of the same instantaneous size as $R(X_k)$ is harmless when it is centered. We now give an online update that performs this centering in the normal-form model itself. The update does not require an independent simulator, a batch of fast loops, or an oracle for $\E[R(X_k)]$. It simply tracks the observed nonlinear term on an intermediate time scale.

Let
\begin{equation}
  \gamma_k=\frac{\gamma_0}{\bar k^c},
  \qquad \gamma_0>0, \label{eq:gamma-stepsize}
\end{equation}
where the exponent $c$ satisfies
\begin{equation}
  a<c<\min\left\{1,\frac{1+s}{2}\right\},
  \qquad s=a(1+\rho). \label{eq:c-condition}
\end{equation}
The interval is non-empty for every $a\in(1/2,1)$ and $\rho\in[0,1]$. Indeed, if $s\ge1$, then $a<1$ gives $a<\min\{1,(1+s)/2\}=1$; if $s<1$, then
\[
  \frac{1+s}{2}-a
  =\frac{1-a+a\rho}{2}>0.
\]
Increasing $k_0$ if necessary, assume also that $\gamma_k\le1$ for all $k$.

The bias-tracked normal-form recursion is
\begin{align}
  X_{k+1} &= (I-\alpha_k A)X_k+\alpha_k\xi_{k+1}, \label{eq:deb-fast}\\
  M_{k+1} &= (1-\gamma_k)M_k+\gamma_k R(X_k), \label{eq:bias-tracker}\\
  \widetilde Y_{k+1}
  &= (I-\beta_k B)\widetilde Y_k
     +\beta_k\big(HX_k+R(X_k)-M_k+\zeta_{k+1}\big). \label{eq:deb-slow}
\end{align}
The variable $M_k\in\R^{d_y}$ is a running estimate of the predictable nonlinear bias. The time-scale ordering
\[
  \frac{\beta_k}{\gamma_k}\to0,
  \qquad
  \frac{\gamma_k}{\alpha_k}\to0
\]
means that $M_k$ evolves faster than the slow recursion but slower than the fast tracking recursion.

\begin{theorem}[Online bias tracking averts the lower bound]
\label{thm:bias-tracking}
Suppose Assumptions~\ref{ass:stability}, \ref{ass:remainder}, and \ref{ass:strength} hold. Let $\gamma_k$ satisfy \eqref{eq:gamma-stepsize}--\eqref{eq:c-condition}, and let $(X_k,M_k,\widetilde Y_k)$ follow \eqref{eq:deb-fast}--\eqref{eq:deb-slow} with deterministic initial conditions. Then there is a finite constant $C$ such that, for every $k\ge1$,
\begin{equation}
  \E\norm{M_k}^2\le C\bar k^{-s},
  \qquad
  \E\norm{\widetilde Y_k}^2\le C\bar k^{-1}. \label{eq:bias-tracking-rate}
\end{equation}
Consequently, in the local normal form, the slower term $\bar k^{-a(1+\rho)}$ is unavoidable for the uncorrected recursion but is removed by the online bias-tracking update.
\end{theorem}

The proof is given in Appendix~\ref{app:bias-tracking-proof}.

\begin{remark}[What the theorem does and does not require]
The theorem uses no additional statistical assumptions beyond those already imposed for Theorem~\ref{thm:upper}. The new ingredient is algorithmic: the normal form exposes $R(X_k)$, and the update \eqref{eq:bias-tracker} tracks this observed quantity. In a general nonlinear TTSA recursion, the analogous coupling error is hidden inside $g(x_k,y_k)-g(x^*(y_k),y_k)$, and $x^*(y_k)$ is typically unknown. Extending the online bias tracker to that setting therefore requires additional structure, such as a local Taylor control variate, an estimator of the reduced drift, or a frozen slow iterate batching construction. The present theorem resolves the algorithmic bias question in the local normal form without such additional access.
\end{remark}

\section{Localized transfer to nonlinear TTSA}
\label{sec:localized-transfer}

The normal-form analysis isolates the mechanism by which local regularity of the slow drift controls the slow rate. We now show how the same mechanism transfers to nonlinear two-time-scale stochastic approximation near a stable equilibrium. The result is local in nature: it applies on time horizons where the iterates remain in a neighborhood in which fast-manifold coordinates and the displayed local error recursions are valid. This includes projected, truncated, or otherwise localized algorithms, as well as analyses combined with an independent stability argument. Any correction terms introduced by the chosen localization mechanism may be included in the perturbations below, provided they satisfy the bounds in Assumption~\ref{ass:nonlinear-reduction}.

Consider the nonlinear two-time-scale recursion
\begin{align}
  x_{k+1} &= x_k+\alpha_k\big(h(x_k,y_k)+M^x_{k+1}\big), \label{eq:general-x}\\
  y_{k+1} &= y_k+\beta_k\big(g(x_k,y_k)+M^y_{k+1}\big), \label{eq:general-y}
\end{align}
where $h:\R^{d_x}\times\R^{d_y}\to\R^{d_x}$ and
$g:\R^{d_x}\times\R^{d_y}\to\R^{d_y}$. Suppose that, in a neighborhood
$\mathcal U$ of the target slow equilibrium $y^*$, there is a map
$\lambda:\mathcal U\to\R^{d_x}$ satisfying
\[
  h(\lambda(y),y)=0,
  \qquad y\in\mathcal U.
\]
The map $\lambda$ is the \emph{fast equilibrium manifold}. For a fixed value of the slow variable $y$, the fast mean field $x\mapsto h(x,y)$ has equilibrium $\lambda(y)$. Thus, if $y$ were frozen, the fast recursion would be expected to track $\lambda(y)$. In the two-time-scale recursion, however, $y_k$ changes slowly over time, so the fast iterate tracks a moving target. This motivates the fast-manifold coordinates
\[
  X_k=x_k-\lambda(y_k),
  \qquad
  Y_k=y_k-y^*.
\]
Here $X_k$ is the genuine fast tracking error, while $Y_k$ is the slow error. The reduced slow drift is the map
\[
  G(y)=g(\lambda(y),y),
\]
and the point $y^*$ is a stable zero of this reduced drift in the local model below.

The movement of the manifold $\lambda(y_k)$ is the main new feature relative to the normal form. Since
\[
  X_{k+1}=x_{k+1}-\lambda(y_{k+1}),
\]
increments of $y_k$ feed back into the fast tracking recursion. In particular, the slow noise $M^y_{k+1}$ produces a perturbation of size $\beta_k$ in the $X$-recursion. The transfer theorem below treats the nonlinear recursion in fast-manifold coordinates as a perturbation of the normal form. The term $p_{k+1}$ represents fast-recursion residuals at the $\alpha_k$ scale, $d_k$ represents predictable perturbations entering through the moving manifold at the $\beta_k$ scale, and $e_k$ represents higher-order slow-drift terms beyond the components $HX_k+R(X_k)$. The assumption is designed so that these perturbations remain below the phase-transition rate after the same Abel-transform cancellation used in the normal-form analysis.

\begin{assumption}[Localized nonlinear reduction]\label{ass:nonlinear-reduction}
There are matrices $A,B,H$, a matrix $C_\zeta$, a map
$R:\R^{d_x}\to\R^{d_y}$, and adapted perturbation processes
$p_{k+1}$, $d_k$, and $e_k$ such that the local error process satisfies
\begin{align}
  X_{k+1}
  &=(I-\alpha_kA)X_k+\alpha_k\xi_{k+1}
    +\alpha_k p_{k+1}
    +\beta_k C_\zeta\zeta_{k+1}
    +\beta_kd_k, \label{eq:perturbed-fast}\\
  Y_{k+1}
  &=(I-\beta_kB)Y_k
    +\beta_k\big(HX_k+R(X_k)+\zeta_{k+1}+e_k\big), \label{eq:perturbed-slow}
\end{align}
where $\xi_{k+1}$ and $\zeta_{k+1}$ are martingale differences with respect to $\mathcal F_k$. The perturbation terms $p_{k+1}$, $d_k$, and $e_k$ are not required to be martingale differences.

Assumption~\ref{ass:stability} holds for $A,B,\xi,\zeta$, and Assumption~\ref{ass:remainder} holds for $R$. With
\[
  s=a(1+\rho),
  \qquad
  \delta=\min\{1,s\},
\]
there are finite constants $K_X,K_p,K_d,K_e$ and an exponent
$\mu_p>\delta/2$ such that, for every $k\ge0$,
\begin{align}
  \E\norm{X_k}^2 &\le K_X\alpha_k,
  &
  \E\norm{X_k}^{2+2\rho} &\le K_X\alpha_k^{1+\rho}, \label{eq:transfer-fast-moments}\\
  \big(\E\norm{p_{k+1}}^2\big)^{1/2}
  &\le K_p\bar k^{-\mu_p}, \label{eq:p-bound}\\
  \big(\E\norm{d_k}^2\big)^{1/2}
  &\le K_d\big(\bar k^{-a/2}+\bar k^{-\delta/2}\big), \label{eq:d-bound}\\
  \big(\E\norm{e_k}^2\big)^{1/2}
  &\le K_e\bar k^{-\delta/2}. \label{eq:e-bound}
\end{align}
Finally, Assumption~\ref{ass:strength} holds.
\end{assumption}

\begin{theorem}[Localized nonlinear transfer]\label{thm:localized-transfer}
Under Assumption~\ref{ass:nonlinear-reduction}, there is a finite constant $C$, depending only on the constants in the assumption and on the fixed problem parameters, such that, for every $k\ge1$,
\begin{equation}
  \E\norm{Y_k}^2
  \le C\left(\bar k^{-1}+\bar k^{-a(1+\rho)}\right). \label{eq:localized-transfer-rate}
\end{equation}
Consequently, the phase transition from Corollary~\ref{cor:phase} holds for every localized nonlinear recursion satisfying the reduction above.
\end{theorem}

The next two propositions verify Assumption~\ref{ass:nonlinear-reduction} from local Taylor structure. They are intended to be combined with any mechanism that supplies the required localized moment bounds, such as projection, truncation, a Lyapunov stability argument, or an existing finite-time nonlinear TTSA estimate. Proposition~\ref{prop:verification} covers a separable local structure in which the leading slow remainder depends only on the fast tracking error. Proposition~\ref{prop:generic-verification} covers generic mixed local terms, at the cost of a fourth-moment bound on the slow error.

\begin{proposition}[Verification from separable local $C^{1,\rho}$ structure]\label{prop:verification}
Consider a localized version of \eqref{eq:general-x}--\eqref{eq:general-y} on a time horizon where $y_k\in\mathcal U$ and the iterates remain in a bounded neighborhood of $(\lambda(y^*),y^*)$. Suppose that the following properties hold on this neighborhood.
\begin{enumerate}[leftmargin=*]
\item The fast equilibrium manifold $\lambda$ is twice continuously differentiable with bounded first and second derivatives.

\item There are matrices $A,B,H$ and a map $R$ satisfying Assumption~\ref{ass:remainder} such that, for all local $u$ and $y$,
\begin{align}
  h(\lambda(y)+u,y)
  &=-Au+P(u,y),
  &
  \norm{P(u,y)}
  &\le L_h\norm{u}^2, \label{eq:h-local}\\
  g(\lambda(y)+u,y)
  &=-B(y-y^*)+Hu+R(u)+E(u,y),
  &
  \norm{E(u,y)}
  &\le L_g\norm{u}^{1+\rho}. \label{eq:g-local}
\end{align}
The matrices $A$ and $B$ satisfy the contraction inequalities in Assumption~\ref{ass:stability}, and the stepsize constants satisfy Assumption~\ref{ass:strength}.

\item The noises $M^x_{k+1}$ and $M^y_{k+1}$ are martingale differences with uniformly bounded fourth conditional moments. Set
\[
  \xi_{k+1}=M^x_{k+1},
  \qquad
  \zeta_{k+1}=M^y_{k+1}.
\]

\item A baseline localized moment estimate is available: for finite constants $K_0,K_1$,
\begin{equation}
  \E\norm{X_k}^4\le K_0\alpha_k^2,
  \qquad
  \E\norm{Y_k}^2\le K_1\bar k^{-a},
  \qquad k\ge0. \label{eq:baseline-localization}
\end{equation}
\end{enumerate}
Then the localized recursion satisfies Assumption~\ref{ass:nonlinear-reduction}. Consequently, the conclusion of Theorem~\ref{thm:localized-transfer} holds.
\end{proposition}

\begin{proposition}[Verification with generic mixed local terms]\label{prop:generic-verification}
Consider a localized version of \eqref{eq:general-x}--\eqref{eq:general-y} on a time horizon where $y_k\in\mathcal U$ and the iterates remain in a bounded neighborhood of $(\lambda(y^*),y^*)$. Suppose that the following properties hold on this neighborhood.
\begin{enumerate}[leftmargin=*]
\item The fast equilibrium manifold $\lambda$ is twice continuously differentiable with bounded first and second derivatives.

\item There are matrices $A,B,H$ and a map $R$ satisfying Assumption~\ref{ass:remainder} such that, for all local $u$ and $y$ with $Y=y-y^*$,
\begin{align}
  h(\lambda(y)+u,y)
  &=-Au+P(u,y),
  &
  \norm{P(u,y)}
  &\le L_h\big(\norm{u}^2+\norm{u}\norm{Y}\big), \label{eq:h-generic}\\
  g(\lambda(y)+u,y)
  &=-BY+Hu+R(u)+E(u,y),
  &
  \norm{E(u,y)}
  &\le L_g\big(\norm{u}^{1+\rho}+\norm{u}\norm{Y}+\norm{Y}^2\big). \label{eq:g-generic}
\end{align}
The matrices $A$ and $B$ satisfy the contraction inequalities in Assumption~\ref{ass:stability}, and the stepsize constants satisfy Assumption~\ref{ass:strength}.

\item The noises $M^x_{k+1}$ and $M^y_{k+1}$ are martingale differences with uniformly bounded fourth conditional moments. Set
\[
  \xi_{k+1}=M^x_{k+1},
  \qquad
  \zeta_{k+1}=M^y_{k+1}.
\]

\item A fourth-moment localized estimate is available: for finite constants $K_0,K_1$,
\begin{equation}
  \E\norm{X_k}^4\le K_0\alpha_k^2,
  \qquad
  \E\norm{Y_k}^4\le K_1\bar k^{-2a},
  \qquad k\ge0. \label{eq:baseline-generic}
\end{equation}
\end{enumerate}
Then the localized recursion satisfies Assumption~\ref{ass:nonlinear-reduction}. Consequently, the conclusion of Theorem~\ref{thm:localized-transfer} holds.
\end{proposition}

\section{Proof ideas}

All proofs are given in full detail in the appendix. We summarize the main mechanisms here. Throughout this discussion, write
\[
  \bar k=k+k_0,
  \qquad
  s=a(1+\rho).
\]
The proof separates three effects in the slow recursion: ordinary slow noise, locally linear leakage from the fast tracking error, and nonlinear predictable bias.

\paragraph{Fast fluctuation scale.}
The fast recursion is stable and is driven by noise through the stepsize $\alpha_k$. A standard Lyapunov calculation therefore gives
\[
  \E\norm{X_k}^2=O(\alpha_k),
  \qquad
  \E\norm{X_k}^4=O(\alpha_k^2).
\]
By interpolation,
\[
  \E\norm{X_k}^{2+2\rho}=O(\alpha_k^{1+\rho}).
\]
Thus the natural size of the fast tracking error is
\[
  \norm{X_k}\approx \sqrt{\alpha_k}\asymp \bar k^{-a/2}.
\]
Consequently, a nonlinear remainder satisfying
\[
  \norm{R(x)}\lesssim \norm{x}^{1+\rho}
\]
has typical magnitude
\[
  \norm{R(X_k)}
  \approx
  \alpha_k^{(1+\rho)/2}
  \asymp
  \bar k^{-a(1+\rho)/2}.
\]
This is the source of the phase-transition exponent.

\paragraph{Unrolling the slow recursion.}
Let
\[
  \Pi^B_{i,k}=\prod_{\ell=i}^{k-1}(I-\beta_\ell B),
  \qquad
  \norm{\Pi^B_{i,k}}\le C\left(\frac{i+k_0}{k+k_0}\right)^r .
\]
Unrolling the slow recursion gives
\[
  Y_k
  =
  \Pi^B_{0,k}Y_0
  +
  \sum_{j<k}\beta_j\Pi^B_{j+1,k}HX_j
  +
  \sum_{j<k}\beta_j\Pi^B_{j+1,k}R(X_j)
  +
  \sum_{j<k}\beta_j\Pi^B_{j+1,k}\zeta_{j+1}.
\]
The initial-condition term decays faster than the rates of interest once the slow contraction is chosen sufficiently strong.

The slow-noise term is standard. Since $\zeta_{j+1}$ is a martingale difference,
\[
  \E\left\|
  \sum_{j<k}\beta_j\Pi^B_{j+1,k}\zeta_{j+1}
  \right\|^2
  \le
  C\sum_{j<k}\beta_j^2
  \left(\frac{j+k_0}{k+k_0}\right)^{2r}
  =
  O(\bar k^{-1}).
\]
This is the intrinsic $k^{-1}$ noise floor of the slow stochastic approximation.

\paragraph{The nonlinear remainder.}
The nonlinear term is bounded directly. By Minkowski's inequality and the fast moment estimate,
\[
\begin{aligned}
  \left(
  \E\left\|
  \sum_{j<k}\beta_j\Pi^B_{j+1,k}R(X_j)
  \right\|^2
  \right)^{1/2}
  &\le
  C\sum_{j<k}\beta_j
  \left(\frac{j+k_0}{k+k_0}\right)^r
  \left(\E\norm{X_j}^{2+2\rho}\right)^{1/2}  \\
  &\le
  C\sum_{j<k}
  \frac{1}{j+k_0}
  \left(\frac{j+k_0}{k+k_0}\right)^r
  (j+k_0)^{-a(1+\rho)/2}.
\end{aligned}
\]
A weighted-convolution estimate gives an $L^2$ bound of order
\[
  O\!\left(\bar k^{-\min\{r,a(1+\rho)/2\}}\right)
\]
up to a logarithm in the borderline case. After squaring, this is always bounded by
\[
  O\!\left(\bar k^{-1}+\bar k^{-a(1+\rho)}\right),
\]
because $r>1/2$. This term is not improved by averaging because it is a predictable drift input rather than a martingale-centered fluctuation.

\paragraph{Locally linear leakage and Abel cancellation.}
The locally linear term $HX_j$ is subtler. A direct estimate using
$\E\norm{X_j}^2=O(\alpha_j)$ would only give the slower $O(k^{-a})$ rate. The key point is that $X_j$ is not an arbitrary error sequence; it is generated by a stable fast recursion. From
\[
  X_{j+1}=(I-\alpha_jA)X_j+\alpha_j\xi_{j+1},
\]
we obtain the identity
\[
  AX_j=\alpha_j^{-1}(X_j-X_{j+1})+\xi_{j+1}.
\]
With $C_H=HA^{-1}$,
\[
\begin{aligned}
  \sum_{j<k}\beta_j\Pi^B_{j+1,k}HX_j
  &=
  \sum_{j<k}\beta_j\Pi^B_{j+1,k}C_H\xi_{j+1} \\
  &\quad+
  \sum_{j<k}\frac{\beta_j}{\alpha_j}
  \Pi^B_{j+1,k}C_H(X_j-X_{j+1}).
\end{aligned}
\]
The first term is again a martingale convolution and has mean square $O(\bar k^{-1})$.

The second term is handled by summation by parts. The boundary coefficients are of order
\[
  \frac{\beta_k}{\alpha_k}\asymp \bar k^{a-1}.
\]
Since $\E\norm{X_k}^2=O(\bar k^{-a})$, the terminal boundary term contributes
\[
  \left(\frac{\beta_k}{\alpha_k}\right)^2\E\norm{X_k}^2
  =
  O(\bar k^{a-2})
  \le
  O(\bar k^{-1}).
\]
The initial boundary is negligible by contraction, and the interior coefficient differences are small enough to give the same rate. This summation-by-parts step is the Abel-transform form of averaged-noise cancellation: locally linear leakage from the fast recursion averages down to the slow $k^{-1}$ noise floor.

Combining the three contributions yields
\[
  \E\norm{Y_k}^2
  \le
  C\left(\bar k^{-1}+\bar k^{-a(1+\rho)}\right).
\]

\paragraph{Why the lower bound matches.}
The lower bound uses a scalar Gaussian normal form. In that example, the fast variable is Gaussian with variance
\[
  \E X_k^2=\Theta(\alpha_k).
\]
Therefore
\[
  \E |X_k|^{1+\rho}
  =
  \Theta\!\left(\alpha_k^{(1+\rho)/2}\right)
  =
  \Theta\!\left(\bar k^{-a(1+\rho)/2}\right).
\]
The slow recursion contains the positive predictable drift input
\[
  \gamma |X_k|^{1+\rho}.
\]
Summing the slow recursion over the window $j\in[k/2,k]$, where the slow transition weights remain uniformly comparable to constants and the input varies only polynomially, gives
\[
  \E Y_k
  \ge
  c\bar k^{-a(1+\rho)/2}.
\]
Hence
\[
  \E Y_k^2
  \ge
  (\E Y_k)^2
  \ge
  c\bar k^{-a(1+\rho)}.
\]
If the slow Gaussian noise level is nonzero, the usual martingale variance calculation gives the additional lower bound $c\bar k^{-1}$. Thus the upper bound is sharp uniformly over the normal-form class.

\paragraph{Centering and debiasing.}
The centering result explains why the obstruction is not the size of the nonlinear term alone. Suppose the nonlinear input is replaced by a martingale-centered fluctuation $\widetilde R_{k+1}$ satisfying
\[
  \E[\widetilde R_{k+1}\mid\mathcal F_k]=0,
  \qquad
  \E\norm{\widetilde R_{k+1}}^2
  \lesssim
  \alpha_k^{1+\rho}.
\]
Then its slow convolution satisfies
\[
  \E\left\|
  \sum_{j<k}\beta_j\Pi^B_{j+1,k}\widetilde R_{j+1}
  \right\|^2
  \le
  C\sum_{j<k}\beta_j^2
  \left(\frac{j+k_0}{k+k_0}\right)^{2r}
  \alpha_j^{1+\rho}
  \le
  C\bar k^{-1}.
\]
Thus a nonlinear term of the same instantaneous magnitude does not slow the slow iterate when it is martingale-centered. The rate loss comes from persistent predictable bias.

\paragraph{Online bias tracking.}
The bias-tracking recursion removes the predictable component without an oracle. Since
\[
  M_{k+1}-M_k=\gamma_k\big(R(X_k)-M_k\big),
\]
the corrected nonlinear input satisfies
\[
  R(X_k)-M_k=\gamma_k^{-1}(M_{k+1}-M_k).
\]
Therefore its slow convolution is
\[
  \sum_{j<k}\beta_j\Pi^B_{j+1,k}\big(R(X_j)-M_j\big)
  =
  \sum_{j<k}\frac{\beta_j}{\gamma_j}\Pi^B_{j+1,k}(M_{j+1}-M_j).
\]
This is another Abel-transform term. The tracker itself satisfies
\[
  \norm{M_k}_{L^2}=O(\bar k^{-s/2}),
\]
because it is an exponentially weighted average of inputs $R(X_j)$ of $L^2$ size $O(\bar j^{-s/2})$. The coefficient $\beta_j/\gamma_j\asymp \bar j^{c-1}$ is small because the tracking time scale is faster than the slow time scale. The condition $c<\min\{1,(1+s)/2\}$ ensures that the boundary and coefficient-variation terms in this Abel transform have $L^2$ size $O(\bar k^{-1/2})$. Thus the tracked nonlinear residual contributes only $O(\bar k^{-1})$ to the slow mean-square error.

\paragraph{Localized nonlinear transfer.}
For nonlinear TTSA, the same argument is applied after changing coordinates to the fast-manifold error
\[
  X_k=x_k-\lambda(y_k),
  \qquad
  Y_k=y_k-y^*.
\]
Here $\lambda(y)$ is the equilibrium of the fast mean field with the slow variable frozen. Because this target moves with $y_k$, the $X$-recursion contains additional perturbations beyond the normal form. Assumption~\ref{ass:nonlinear-reduction} collects these perturbations as $p_{k+1}$, $d_k$, and $e_k$.

The proof of the transfer theorem repeats the normal-form decomposition. The slow error is unrolled into the same three main terms, plus perturbation convolutions. The direct slow perturbation $e_k$ is controlled by the same convolution estimate used for $R(X_k)$. The perturbations in the fast recursion are handled through the Abel transform for the locally linear term $HX_k$. Terms entering the fast recursion at the $\beta_k$ scale acquire an additional factor $\beta_k/\alpha_k$ after the transform; since $\beta_k/\alpha_k\asymp \bar k^{a-1}$ and $a<1$, these terms remain below the phase-transition rate under the stated moment bounds. The verification propositions show that the required perturbation bounds follow from local Taylor expansions of $h$ and $g$, bounded derivatives of the manifold $\lambda$, and the corresponding localized moment estimates.

\section{Conclusion and Discussion}

This paper identifies a regularity-dependent boundary for the uncorrected normal-form recursion in nonlinear two-time-scale stochastic approximation. Locally linear fast-error leakage averages down to the intrinsic $O(k^{-1})$ slow noise floor, while a nonlinear remainder of order $1+\rho$ creates a predictable bias at the fast fluctuation scale. This yields the sharp rate
\[
  \E\norm{Y_k}^2
  \le
  C\left(k^{-1}+k^{-a(1+\rho)}\right),
\]
with a matching scalar lower bound. Thus, for the uncorrected normal-form recursion, the decoupled $O(k^{-1})$ rate is guaranteed uniformly precisely when $a(1+\rho)\ge1$.

The lower bound also identifies how to beat it. It is not an information-theoretic impossibility result for the normal-form problem class; it is the bias of a particular update rule. The online tracker
\[
  M_{k+1}=M_k+\gamma_k\big(R(X_k)-M_k\big)
\]
estimates the predictable nonlinear component on an intermediate time scale, and subtracting $M_k$ from the slow update restores
\[
  \E\norm{\widetilde Y_k}^2=O(k^{-1})
\]
for every $\rho\in[0,1]$. The proof uses a second Abel transform, now applied to $R(X_k)-M_k=\gamma_k^{-1}(M_{k+1}-M_k)$.

The localized transfer results show that the same phase-transition mechanism persists for nonlinear TTSA in fast-manifold coordinates under the stated localized moment estimates. Extending the online bias tracker beyond the normal form is a natural next step, but it requires a way to observe or estimate the hidden coupling error $g(x_k,y_k)-g(x^*(y_k),y_k)$. Other directions include simulator-based batching or multilevel debiasing for general TTSA, high-probability analogues, and non-asymptotic central limit theorems with the nonlinear bias explicitly centered.
\newpage
\bibliographystyle{unsrt}
\bibliography{refs}
\newpage
\appendix

\section{Deterministic estimates}\label{app:deterministic}

Throughout the appendix, $C$ denotes a finite positive constant whose value may change from line to line but depends only on the fixed problem parameters. It never depends on $k$ or on summation indices. Recall that $\bar k=k+k_0$.

For $0\le i\le k$, define the slow transition matrix
\begin{equation}
  \Pi^B_{i,k}\coloneqq
  \begin{cases}
    (I-\beta_{k-1}B)(I-\beta_{k-2}B)\cdots(I-\beta_iB), & i<k,\\
    I, & i=k.
  \end{cases} \label{eq:B-transition}
\end{equation}
The order of multiplication is irrelevant for the estimates because every factor is a polynomial in the same matrix $B$, but the displayed convention matches the recursion.

\begin{lemma}[Polynomial decay of the slow transition]\label{lem:transition}
Under Assumption \ref{ass:stability}, for all $0\le i\le k$,
\begin{equation}
  \norm{\Pi^B_{i,k}}
  \le \left(\frac{\bar i}{\bar k}\right)^r,
  \qquad r=b\beta_0. \label{eq:transition-bound}
\end{equation}
\end{lemma}

\begin{proof}
If $i=k$, then \eqref{eq:transition-bound} is equality. Assume $i<k$. By \eqref{eq:matrix-contraction},
\[
  \norm{\Pi^B_{i,k}}
  \le \prod_{\ell=i}^{k-1}(1-b\beta_\ell)
  \le \exp\left(-b\sum_{\ell=i}^{k-1}\beta_\ell\right),
\]
where we used $1-u\le e^{-u}$ for $u\ge0$. Since $\beta_\ell=\beta_0/(\ell+k_0)$,
\[
  \sum_{\ell=i}^{k-1}\beta_\ell
  =\beta_0\sum_{\ell=i}^{k-1}\frac1{\ell+k_0}
  \ge \beta_0\int_i^k \frac{dt}{t+k_0}
  =\beta_0\log\left(\frac{\bar k}{\bar i}\right).
\]
Combining the two inequalities gives
\[
  \norm{\Pi^B_{i,k}}
  \le \exp\left(-b\beta_0\log\left(\frac{\bar k}{\bar i}\right)\right)
  =\left(\frac{\bar i}{\bar k}\right)^{b\beta_0}.
\]
This proves the lemma.
\end{proof}

\begin{lemma}[Weighted convolution estimates]\label{lem:weighted-sums}
Let $u\ge0$ and $r>0$. Then, for all $k\ge1$,
\begin{equation}
  \sum_{j=0}^{k-1}\beta_j\norm{\Pi^B_{j+1,k}}\bar j^{-u}
  \le
  C
  \begin{cases}
    \bar k^{-u}, & u<r,\\
    \bar k^{-r}\log(\bar k), & u=r,\\
    \bar k^{-r}, & u>r.
  \end{cases} \label{eq:weighted-first}
\end{equation}
Consequently, if $r>1/2$, then
\begin{equation}
  \left(
  \sum_{j=0}^{k-1}\beta_j\norm{\Pi^B_{j+1,k}}\bar j^{-u}
  \right)^2
  \le C\left(\bar k^{-1}+\bar k^{-2u}\right). \label{eq:weighted-first-squared}
\end{equation}
If $r>1/2$, then, for all $k\ge1$,
\begin{equation}
  \sum_{j=0}^{k-1}\beta_j^2\norm{\Pi^B_{j+1,k}}^2
  \le C\bar k^{-1}. \label{eq:weighted-second}
\end{equation}
\end{lemma}

\begin{proof}
For \eqref{eq:weighted-first}, Lemma \ref{lem:transition} gives
\[
\begin{aligned}
  \sum_{j=0}^{k-1}\beta_j\norm{\Pi^B_{j+1,k}}\bar j^{-u}
  &\le \beta_0\sum_{j=0}^{k-1}\bar j^{-1}\left(\frac{\overline{j+1}}{\bar k}\right)^r\bar j^{-u}.
\end{aligned}
\]
Because $k_0\ge2$, $\overline{j+1}=\bar j+1\le 2\bar j$. Therefore
\[
  \left(\frac{\overline{j+1}}{\bar k}\right)^r\le 2^r\bar k^{-r}\bar j^r,
\]
and hence
\begin{equation}
  \sum_{j=0}^{k-1}\beta_j\norm{\Pi^B_{j+1,k}}\bar j^{-u}
  \le C\bar k^{-r}\sum_{j=0}^{k-1}\bar j^{r-u-1}. \label{eq:first-sum-reduced}
\end{equation}
If $r-u>0$, the sum in \eqref{eq:first-sum-reduced} is at most $C\bar k^{r-u}$. If $r-u=0$, it is at most $C\log(\bar k)$. If $r-u<0$, it is bounded by a constant. These three cases prove \eqref{eq:weighted-first}.

We next prove \eqref{eq:weighted-first-squared}. If $u<r$, then \eqref{eq:weighted-first} gives $C\bar k^{-2u}$. If $u>r$, then it gives $C\bar k^{-2r}\le C\bar k^{-1}$ because $r>1/2$. If $u=r$, then $2u=2r>1$, and
\[
  \bar k^{-2r}(\log\bar k)^2\le C\bar k^{-1}.
\]
Thus \eqref{eq:weighted-first-squared} holds in all cases.

For \eqref{eq:weighted-second}, again by Lemma \ref{lem:transition},
\[
\begin{aligned}
  \sum_{j=0}^{k-1}\beta_j^2\norm{\Pi^B_{j+1,k}}^2
  &\le C\bar k^{-2r}\sum_{j=0}^{k-1}\bar j^{2r-2}.
\end{aligned}
\]
Since $2r-2>-1$ is equivalent to $r>1/2$, the sum is at most $C\bar k^{2r-1}$. Therefore the last display is bounded by $C\bar k^{-1}$.
\end{proof}

\begin{lemma}[Abel coefficient bounds]\label{lem:abel-coefficients}
Let $C_H=HA^{-1}$ and define, for $0\le j\le k-1$,
\begin{equation}
  Q_{j,k}\coloneqq \frac{\beta_j}{\alpha_j}\Pi^B_{j+1,k}C_H. \label{eq:Q-def}
\end{equation}
Then
\begin{align}
  \norm{Q_{j,k}} &\le C\bar k^{-r}\bar j^{r+a-1}, \label{eq:Q-size}\\
  \norm{Q_{j,k}-Q_{j-1,k}} &\le C\bar k^{-r}\bar j^{r+a-2},
  \qquad 1\le j\le k-1. \label{eq:Q-variation}
\end{align}
\end{lemma}

\begin{proof}
First,
\[
  \frac{\beta_j}{\alpha_j}
  =\frac{\beta_0}{\alpha_0}\bar j^{a-1}.
\]
Using Lemma \ref{lem:transition},
\[
  \norm{Q_{j,k}}
  \le \frac{\beta_0}{\alpha_0}\norm{C_H}\bar j^{a-1}
      \left(\frac{\overline{j+1}}{\bar k}\right)^r
  \le C\bar k^{-r}\bar j^{r+a-1},
\]
which proves \eqref{eq:Q-size}.

For the variation bound, write
\[
  q_j\coloneqq \frac{\beta_j}{\alpha_j}=\frac{\beta_0}{\alpha_0}\bar j^{a-1}.
\]
Then
\begin{align*}
  Q_{j,k}-Q_{j-1,k}
  &=\left(q_j\Pi^B_{j+1,k}-q_{j-1}\Pi^B_{j,k}\right)C_H\\
  &=\left((q_j-q_{j-1})\Pi^B_{j+1,k}\right)C_H
    +q_{j-1}\left(\Pi^B_{j+1,k}-\Pi^B_{j,k}\right)C_H.
\end{align*}
Because $a-1<0$, the map $t\mapsto(t+k_0)^{a-1}$ is differentiable with derivative of magnitude $(1-a)(t+k_0)^{a-2}$. The mean-value theorem gives
\begin{equation}
  |q_j-q_{j-1}|\le C\bar j^{a-2}. \label{eq:q-diff}
\end{equation}
Moreover, from the definition of the transition matrix,
\[
  \Pi^B_{j,k}=\Pi^B_{j+1,k}(I-\beta_jB),
\]
so
\begin{equation}
  \Pi^B_{j+1,k}-\Pi^B_{j,k}=\beta_j\Pi^B_{j+1,k}B. \label{eq:transition-diff}
\end{equation}
Combining \eqref{eq:q-diff}, \eqref{eq:transition-diff}, Lemma \ref{lem:transition}, and $q_{j-1}\le C\bar j^{a-1}$ yields
\[
\begin{aligned}
  \norm{Q_{j,k}-Q_{j-1,k}}
  &\le C\bar j^{a-2}\left(\frac{\overline{j+1}}{\bar k}\right)^r
     +C\bar j^{a-1}\beta_j\left(\frac{\overline{j+1}}{\bar k}\right)^r\\
  &\le C\bar j^{a-2}\bar k^{-r}\bar j^r
     +C\bar j^{a-1}\bar j^{-1}\bar k^{-r}\bar j^r\\
  &\le C\bar k^{-r}\bar j^{r+a-2}.
\end{aligned}
\]
This proves \eqref{eq:Q-variation}.
\end{proof}

\section{Proof of the upper bound}\label{app:upper}

We now prove Theorem \ref{thm:upper}. The proof is divided into three lemmas and then assembled.

\begin{lemma}[Fast moments]\label{lem:fast-moments}
Under Assumption \ref{ass:stability}, there exists $C<\infty$ such that, for every $k\ge0$,
\begin{equation}
  \E\norm{X_k}^2\le C\alpha_k,\qquad
  \E\norm{X_k}^4\le C\alpha_k^2. \label{eq:fast-second-fourth}
\end{equation}
Consequently, for every $\rho\in[0,1]$,
\begin{equation}
  \E\norm{X_k}^{2+2\rho}\le C\alpha_k^{1+\rho}. \label{eq:fast-interp}
\end{equation}
\end{lemma}

\begin{proof}
Let $D_k=I-\alpha_kA$. By Assumption \ref{ass:stability}, $\norm{D_k}\le1-\lambda\alpha_k$.

\paragraph{Second moment.}
Conditioning on $\mathcal F_k$ and using $\E[\xi_{k+1}\mid\mathcal F_k]=0$,
\[
\begin{aligned}
  \E[\norm{X_{k+1}}^2\mid\mathcal F_k]
  &=\E[\norm{D_kX_k+\alpha_k\xi_{k+1}}^2\mid\mathcal F_k]\\
  &=\norm{D_kX_k}^2
    +2\alpha_k\ip{D_kX_k}{\E[\xi_{k+1}\mid\mathcal F_k]}
    +\alpha_k^2\E[\norm{\xi_{k+1}}^2\mid\mathcal F_k]\\
  &\le (1-\lambda\alpha_k)^2\norm{X_k}^2+\sigma_4^2\alpha_k^2.
\end{aligned}
\]
Since $\lambda\alpha_k\le1$ after increasing $k_0$ if necessary, $(1-\lambda\alpha_k)^2\le1-\lambda\alpha_k$. Thus
\begin{equation}
  m_{k+1}\le (1-\lambda\alpha_k)m_k+\sigma_4^2\alpha_k^2,\qquad
  m_k\coloneqq \E\norm{X_k}^2. \label{eq:m-recursion}
\end{equation}
We prove $m_k\le C_2\alpha_k$ by induction. Since $X_0$ is deterministic and $\alpha_0$ in the display below denotes the fixed numerator in the stepsize formula, choose $C_2$ so large that $m_0\le C_2\alpha_0 k_0^{-a}$. Because $\alpha_k=\alpha_0\bar k^{-a}$,
\[
  \frac{\alpha_{k+1}}{\alpha_k}=\left(\frac{\bar k}{\bar k+1}\right)^a
  =\left(1+\frac1{\bar k}\right)^{-a}.
\]
The inequality $(1+t)^{-a}\ge1-at$ for $t\ge0$ gives
\begin{equation}
  \alpha_{k+1}\ge \alpha_k\left(1-\frac{a}{\bar k}\right). \label{eq:alpha-lower}
\end{equation}
Since $a<1$, we have $\bar k^{-1}=o(\alpha_k)$ because $\alpha_k=\alpha_0\bar k^{-a}$. Hence, after increasing $k_0$ if necessary,
\begin{equation}
  \frac{a}{\bar k}\le \frac{\lambda}{2}\alpha_k,
  \qquad k\ge0. \label{eq:k0-condition-second}
\end{equation}
Combining \eqref{eq:alpha-lower} and \eqref{eq:k0-condition-second},
\begin{equation}
  \alpha_{k+1}\ge \alpha_k\left(1-\frac{\lambda}{2}\alpha_k\right). \label{eq:alpha-rec-compare}
\end{equation}
Assume $m_k\le C_2\alpha_k$. From \eqref{eq:m-recursion},
\[
  m_{k+1}\le C_2\alpha_k(1-\lambda\alpha_k)+\sigma_4^2\alpha_k^2
  =C_2\alpha_k\left(1-\lambda\alpha_k+\frac{\sigma_4^2}{C_2}\alpha_k\right).
\]
Choose $C_2\ge 2\sigma_4^2/\lambda$. Then
\[
  m_{k+1}\le C_2\alpha_k\left(1-\frac{\lambda}{2}\alpha_k\right)
  \le C_2\alpha_{k+1},
\]
where the last inequality uses \eqref{eq:alpha-rec-compare}. This completes the induction and proves $m_k\le C_2\alpha_k$.

\paragraph{Fourth moment.}
Let $u=D_kX_k$ and $v=\alpha_k\xi_{k+1}$. The identity
\[
  \norm{u+v}^4=(\norm{u}^2+2\ip{u}{v}+\norm{v}^2)^2
\]
expands as
\begin{align*}
  \norm{u+v}^4
  &=\norm{u}^4+4\norm{u}^2\ip{u}{v}
    +2\norm{u}^2\norm{v}^2+4\ip{u}{v}^2\\
  &\qquad +4\ip{u}{v}\norm{v}^2+\norm{v}^4.
\end{align*}
Taking conditional expectation, the term $4\norm{u}^2\ip{u}{v}$ vanishes because $\E[v\mid\mathcal F_k]=0$. For the remaining terms, using $\norm{u}\le\norm{X_k}$ and \eqref{eq:moment-noise},
\begin{align*}
  \E[2\norm{u}^2\norm{v}^2\mid\mathcal F_k]
  &\le 2\sigma_4^2\alpha_k^2\norm{X_k}^2,\\
  \E[4\ip{u}{v}^2\mid\mathcal F_k]
  &\le 4\norm{u}^2\E[\norm{v}^2\mid\mathcal F_k]
   \le 4\sigma_4^2\alpha_k^2\norm{X_k}^2,\\
  \E[\norm{v}^4\mid\mathcal F_k]
  &\le \sigma_4^4\alpha_k^4.
\end{align*}
For the mixed cubic term, first use Cauchy's inequality and Jensen's inequality:
\[
  \left|\E[4\ip{u}{v}\norm{v}^2\mid\mathcal F_k]\right|
  \le 4\norm{u}\E[\norm{v}^3\mid\mathcal F_k]
  \le 4\sigma_4^3\alpha_k^3\norm{X_k}.
\]
We now apply Young's inequality in the form
\[
  c_1\alpha_k^3 t
  \le \frac{\lambda}{4}\alpha_k t^4+C\alpha_k^{11/3},
  \qquad t\ge0, \label{eq:young-fourth}
\]
which follows from $ab\le \varepsilon a^4 + C\varepsilon^{-1/3}b^{4/3}$ by taking $a=t$, $b=c_1\alpha_k^3$, and $\varepsilon=(\lambda/4)\alpha_k$. Since $\alpha_k\le1$, $\alpha_k^{11/3}\le\alpha_k^3$. Therefore
\begin{equation}
  \left|\E[4\ip{u}{v}\norm{v}^2\mid\mathcal F_k]\right|
  \le \frac{\lambda}{4}\alpha_k\norm{X_k}^4+C\alpha_k^3. \label{eq:mixed-cubic}
\end{equation}
Finally,
\[
  \norm{u}^4\le(1-\lambda\alpha_k)^4\norm{X_k}^4.
\]
For $\lambda\alpha_k\le1/4$, the elementary inequality
\[
  (1-x)^4\le 1-\frac{15}{8}x,\qquad 0\le x\le\frac14,
\]
with $x=\lambda\alpha_k$ gives
\[
  \norm{u}^4\le \left(1-\frac{15\lambda}{8}\alpha_k\right)\norm{X_k}^4.
\]
Combining this estimate with \eqref{eq:mixed-cubic} and the remaining second- and fourth-order noise terms gives
\[
  \E[\norm{X_{k+1}}^4\mid\mathcal F_k]
  \le \left(1-\frac{13\lambda}{8}\alpha_k\right)\norm{X_k}^4
       +C\alpha_k^2\norm{X_k}^2+C\alpha_k^3.
\]
Absorbing constants and taking expectation,
\begin{equation}
  n_{k+1}\le(1-c_4\alpha_k)n_k+C\alpha_k^2m_k+C\alpha_k^3,
  \qquad n_k\coloneqq\E\norm{X_k}^4, \label{eq:n-recursion-pre}
\end{equation}
for $c_4=13\lambda/8$ replaced by any smaller positive constant. The already-proved second moment bound gives $m_k\le C\alpha_k$, hence
\begin{equation}
  n_{k+1}\le(1-c_4\alpha_k)n_k+C\alpha_k^3. \label{eq:n-recursion}
\end{equation}
We prove $n_k\le C_4\alpha_k^2$ by induction. Choose $C_4$ so large that $n_0\le C_4(\alpha_0 k_0^{-a})^2$. Since
\[
  \frac{\alpha_{k+1}^2}{\alpha_k^2}=\left(1+\frac1{\bar k}\right)^{-2a}
  \ge1-\frac{2a}{\bar k},
\]
and since $\bar k^{-1}=o(\alpha_k)$, after increasing $k_0$ if necessary,
\begin{equation}
  \alpha_{k+1}^2\ge \alpha_k^2\left(1-\frac{c_4}{2}\alpha_k\right). \label{eq:alpha2-compare}
\end{equation}
Assume $n_k\le C_4\alpha_k^2$. From \eqref{eq:n-recursion},
\[
  n_{k+1}\le C_4\alpha_k^2(1-c_4\alpha_k)+C\alpha_k^3
  =C_4\alpha_k^2\left(1-c_4\alpha_k+\frac{C}{C_4}\alpha_k\right).
\]
Choosing $C_4\ge 2C/c_4$ gives
\[
  n_{k+1}\le C_4\alpha_k^2\left(1-\frac{c_4}{2}\alpha_k\right)
  \le C_4\alpha_{k+1}^2.
\]
This completes the induction and proves the fourth moment bound.

\paragraph{Interpolation.}
Let $p=2+2\rho\in[2,4]$. By Lyapunov's inequality,
\[
  \E\norm{X_k}^p\le \left(\E\norm{X_k}^4\right)^{p/4}
  \le C\left(\alpha_k^2\right)^{p/4}
  =C\alpha_k^{p/2}
  =C\alpha_k^{1+\rho}.
\]
This proves \eqref{eq:fast-interp} and the lemma.
\end{proof}

\begin{lemma}[Averaging of the locally linear fast leakage]\label{lem:linear-leakage}
Under Assumptions \ref{ass:stability} and \ref{ass:strength},
\begin{equation}
  \E\left\|\sum_{j=0}^{k-1}\beta_j\Pi^B_{j+1,k}HX_j\right\|^2
  \le C\bar k^{-1}. \label{eq:linear-leakage-bound}
\end{equation}
\end{lemma}

\begin{proof}
Let $C_H=HA^{-1}$. From the fast recursion,
\[
  X_{j+1}=X_j-\alpha_jAX_j+\alpha_j\xi_{j+1}.
\]
Rearranging,
\begin{equation}
  AX_j=\alpha_j^{-1}(X_j-X_{j+1})+\xi_{j+1}. \label{eq:poisson-identity}
\end{equation}
Multiplying by $HA^{-1}=C_H$ gives
\begin{equation}
  HX_j=C_H\alpha_j^{-1}(X_j-X_{j+1})+C_H\xi_{j+1}. \label{eq:H-identity}
\end{equation}
Therefore
\begin{align}
  \sum_{j=0}^{k-1}\beta_j\Pi^B_{j+1,k}HX_j
  &=M_k+T_k, \label{eq:S-M-T}\\
  M_k&\coloneqq\sum_{j=0}^{k-1}\beta_j\Pi^B_{j+1,k}C_H\xi_{j+1}, \label{eq:M-def}\\
  T_k&\coloneqq\sum_{j=0}^{k-1}Q_{j,k}(X_j-X_{j+1}), \label{eq:T-def}
\end{align}
where $Q_{j,k}$ is defined in \eqref{eq:Q-def}.

We first bound $M_k$. For $i<j$, the random vector $\xi_{i+1}$ is $\mathcal F_j$-measurable and $\E[\xi_{j+1}\mid\mathcal F_j]=0$. Hence cross terms vanish. Thus
\[
\begin{aligned}
  \E\norm{M_k}^2
  &=\sum_{j=0}^{k-1}\beta_j^2\E\norm{\Pi^B_{j+1,k}C_H\xi_{j+1}}^2\\
  &\le \norm{C_H}^2\sigma_4^2\sum_{j=0}^{k-1}\beta_j^2\norm{\Pi^B_{j+1,k}}^2
  \le C\bar k^{-1},
\end{aligned}
\]
where the last step uses Lemma \ref{lem:weighted-sums} and $r>1/2$.

We now bound $T_k$. Abel summation gives
\begin{equation}
  T_k=Q_{0,k}X_0-Q_{k-1,k}X_k+
      \sum_{j=1}^{k-1}(Q_{j,k}-Q_{j-1,k})X_j. \label{eq:abel}
\end{equation}
Indeed, expanding the right-hand side produces
\[
  Q_{0,k}X_0+
  \sum_{j=1}^{k-1}Q_{j,k}X_j-
  \sum_{j=1}^{k-1}Q_{j-1,k}X_j-Q_{k-1,k}X_k
  =\sum_{j=0}^{k-1}Q_{j,k}(X_j-X_{j+1}).
\]
We treat the three terms in \eqref{eq:abel} separately.

By Lemma \ref{lem:abel-coefficients},
\[
  \norm{Q_{0,k}X_0}^2\le C\bar k^{-2r}\norm{X_0}^2\le C\bar k^{-1},
\]
because $r>1/2$. For the terminal boundary term, Lemma \ref{lem:abel-coefficients} and Lemma \ref{lem:fast-moments} imply
\[
\begin{aligned}
  \E\norm{Q_{k-1,k}X_k}^2
  &\le C\bar k^{2a-2}\E\norm{X_k}^2
   \le C\bar k^{2a-2}\alpha_k
   \le C\bar k^{a-2}.
\end{aligned}
\]
Since $a<1$, $\bar k^{a-2}\le C\bar k^{-1}$.

It remains to bound the interior sum. By the triangle inequality in $L^2$,
\[
\begin{aligned}
  \left(\E\left\|\sum_{j=1}^{k-1}(Q_{j,k}-Q_{j-1,k})X_j\right\|^2\right)^{1/2}
  &\le \sum_{j=1}^{k-1}\norm{Q_{j,k}-Q_{j-1,k}}\left(\E\norm{X_j}^2\right)^{1/2}\\
  &\le C\bar k^{-r}\sum_{j=1}^{k-1}\bar j^{r+a-2}\alpha_j^{1/2}\\
  &\le C\bar k^{-r}\sum_{j=1}^{k-1}\bar j^{r+a-2-a/2}\\
  &=C\bar k^{-r}\sum_{j=1}^{k-1}\bar j^{r+a/2-2}.
\end{aligned}
\]
If $r+a/2-2<-1$, then the sum is bounded by a constant, and the squared expression is at most $C\bar k^{-2r}\le C\bar k^{-1}$. If $r+a/2-2=-1$, then the sum is at most $C\log(\bar k)$, and the squared expression is at most $C\bar k^{-2r}(\log \bar k)^2\le C\bar k^{-1}$ because $r>1/2$. If $r+a/2-2>-1$, then the sum is at most $C\bar k^{r+a/2-1}$, and the squared expression is at most
\[
  C\bar k^{-2r}\bar k^{2r+a-2}=C\bar k^{a-2}\le C\bar k^{-1}.
\]
Thus the interior term also has second moment at most $C\bar k^{-1}$.

Using $(u+v)^2\le2u^2+2v^2$ in \eqref{eq:S-M-T}, and using the bounds for $M_k$ and $T_k$, proves \eqref{eq:linear-leakage-bound}.
\end{proof}

\begin{lemma}[Nonlinear and slow-noise convolutions]\label{lem:nonlinear-slow-noise}
Under Assumptions \ref{ass:stability}, \ref{ass:remainder}, and \ref{ass:strength},
\begin{align}
  \E\left\|\sum_{j=0}^{k-1}\beta_j\Pi^B_{j+1,k}R(X_j)\right\|^2
  &\le C\left(\bar k^{-1}+\bar k^{-s}\right), \label{eq:R-bound}\\
  \E\left\|\sum_{j=0}^{k-1}\beta_j\Pi^B_{j+1,k}\zeta_{j+1}\right\|^2
  &\le C\bar k^{-1}. \label{eq:zeta-bound}
\end{align}
\end{lemma}

\begin{proof}
For the nonlinear term, the triangle inequality in $L^2$, Assumption \ref{ass:remainder}, and Lemma \ref{lem:fast-moments} give
\[
\begin{aligned}
  &\left(\E\left\|\sum_{j=0}^{k-1}\beta_j\Pi^B_{j+1,k}R(X_j)\right\|^2\right)^{1/2}\\
  &\qquad\le \sum_{j=0}^{k-1}\beta_j\norm{\Pi^B_{j+1,k}}
       \left(\E\norm{R(X_j)}^2\right)^{1/2}\\
  &\qquad\le L_R\sum_{j=0}^{k-1}\beta_j\norm{\Pi^B_{j+1,k}}
       \left(\E\norm{X_j}^{2+2\rho}\right)^{1/2}\\
  &\qquad\le C\sum_{j=0}^{k-1}\beta_j\norm{\Pi^B_{j+1,k}}\alpha_j^{(1+\rho)/2}\\
  &\qquad\le C\sum_{j=0}^{k-1}\beta_j\norm{\Pi^B_{j+1,k}}\bar j^{-s/2}.
\end{aligned}
\]
By \eqref{eq:weighted-first-squared} with $u=s/2$,
\[
  \E\left\|\sum_{j=0}^{k-1}\beta_j\Pi^B_{j+1,k}R(X_j)\right\|^2
  \le C\left(\bar k^{-1}+\bar k^{-s}\right).
\]
This proves \eqref{eq:R-bound}.

For the slow-noise convolution, martingale cross terms vanish exactly as in the proof of Lemma \ref{lem:linear-leakage}. Hence
\[
\begin{aligned}
  \E\left\|\sum_{j=0}^{k-1}\beta_j\Pi^B_{j+1,k}\zeta_{j+1}\right\|^2
  &=\sum_{j=0}^{k-1}\beta_j^2\E\norm{\Pi^B_{j+1,k}\zeta_{j+1}}^2\\
  &\le \sigma_4^2\sum_{j=0}^{k-1}\beta_j^2\norm{\Pi^B_{j+1,k}}^2
  \le C\bar k^{-1},
\end{aligned}
\]
where the last inequality is \eqref{eq:weighted-second}. This proves \eqref{eq:zeta-bound}.
\end{proof}

\begin{proof}[Proof of Theorem \ref{thm:upper}]
Unroll the slow recursion \eqref{eq:slow} to obtain
\begin{align*}
  Y_k
  &=\Pi^B_{0,k}Y_0
    +\sum_{j=0}^{k-1}\beta_j\Pi^B_{j+1,k}HX_j
    +\sum_{j=0}^{k-1}\beta_j\Pi^B_{j+1,k}R(X_j)
    +\sum_{j=0}^{k-1}\beta_j\Pi^B_{j+1,k}\zeta_{j+1}.
\end{align*}
Using $(a_1+a_2+a_3+a_4)^2\le4(a_1^2+a_2^2+a_3^2+a_4^2)$ for nonnegative scalars,
\begin{align*}
  \E\norm{Y_k}^2
  &\le4\norm{\Pi^B_{0,k}Y_0}^2
    +4\E\left\|\sum_{j=0}^{k-1}\beta_j\Pi^B_{j+1,k}HX_j\right\|^2\\
  &\quad+4\E\left\|\sum_{j=0}^{k-1}\beta_j\Pi^B_{j+1,k}R(X_j)\right\|^2
    +4\E\left\|\sum_{j=0}^{k-1}\beta_j\Pi^B_{j+1,k}\zeta_{j+1}\right\|^2.
\end{align*}
The initial term satisfies, by Lemma \ref{lem:transition},
\[
  \norm{\Pi^B_{0,k}Y_0}^2\le C\bar k^{-2r}\norm{Y_0}^2.
\]
Because $r>1/2$, we have $2r>1$, and therefore
\[
  \bar k^{-2r}\le C\bar k^{-1}\le C(\bar k^{-1}+\bar k^{-s}).
\]
The locally linear term is bounded by Lemma \ref{lem:linear-leakage}; the nonlinear term and slow-noise term are bounded by Lemma \ref{lem:nonlinear-slow-noise}. Combining these estimates gives
\[
  \E\norm{Y_k}^2\le C\left(\bar k^{-1}+\bar k^{-s}\right).
\]
The moment estimates for $X_k$ are exactly Lemma \ref{lem:fast-moments}. This completes the proof.
\end{proof}

\section{Proof of the centering principle}\label{app:centering-proof}

\begin{proof}[Proof of Proposition~\ref{prop:centering}]
Unrolling \eqref{eq:centered-slow} gives
\begin{align*}
  Y_k
  &=\Pi^B_{0,k}Y_0
    +\sum_{j=0}^{k-1}\beta_j\Pi^B_{j+1,k}HX_j
    +\sum_{j=0}^{k-1}\beta_j\Pi^B_{j+1,k}m_j\\
  &\qquad
    +\sum_{j=0}^{k-1}\beta_j\Pi^B_{j+1,k}\nu_{j+1}
    +\sum_{j=0}^{k-1}\beta_j\Pi^B_{j+1,k}\zeta_{j+1}.
\end{align*}
The initial term satisfies $\norm{\Pi^B_{0,k}Y_0}^2\le C\bar k^{-2r}\le C\bar k^{-1}$ because $r>1/2$. Lemma~\ref{lem:linear-leakage} gives
\[
  \E\left\|\sum_{j=0}^{k-1}\beta_j\Pi^B_{j+1,k}HX_j\right\|^2\le C\bar k^{-1}.
\]
The slow martingale noise is bounded as in Lemma~\ref{lem:nonlinear-slow-noise}:
\[
  \E\left\|\sum_{j=0}^{k-1}\beta_j\Pi^B_{j+1,k}\zeta_{j+1}\right\|^2\le C\bar k^{-1}.
\]

We next treat the predictable component $m_j$. By the $L^2$ triangle inequality, \eqref{eq:m-assumption}, and Lemma~\ref{lem:transition},
\begin{align*}
  \left(\E\left\|\sum_{j=0}^{k-1}\beta_j\Pi^B_{j+1,k}m_j\right\|^2\right)^{1/2}
  &\le C\bar k^{-r}\sum_{j=0}^{k-1}\bar j^{r-1}\big(\E\norm{m_j}^2\big)^{1/2}\\
  &\le C\bar k^{-r}\sum_{j=0}^{k-1}\bar j^{r-q-1}.
\end{align*}
If $q<r$, the last sum is at most $C\bar k^{r-q}$; if $q=r$, it is at most $C\log(\bar k)$; if $q>r$, it is bounded by a constant. Therefore
\begin{equation}
  \E\left\|\sum_{j=0}^{k-1}\beta_j\Pi^B_{j+1,k}m_j\right\|^2
  \le C\Psi_q(k)^2. \label{eq:m-centered-bound}
\end{equation}

It remains to control the martingale-centered nonlinear fluctuation. Since $\nu_{j+1}$ is a martingale difference and the matrices $\beta_j\Pi^B_{j+1,k}$ are $\mathcal F_j$-measurable deterministic coefficients, cross terms vanish. Hence, using \eqref{eq:nu-assumption},
\begin{align*}
  \E\left\|\sum_{j=0}^{k-1}\beta_j\Pi^B_{j+1,k}\nu_{j+1}\right\|^2
  &\le C\sum_{j=0}^{k-1}\beta_j^2\norm{\Pi^B_{j+1,k}}^2\alpha_j^{1+\rho}\\
  &\le C\bar k^{-2r}\sum_{j=0}^{k-1}\bar j^{2r-2-s},
\end{align*}
where $s=a(1+\rho)$. If $2r-2-s<-1$, the sum is bounded, and the display is at most $C\bar k^{-2r}\le C\bar k^{-1}$. If $2r-2-s=-1$, the display is at most $C\bar k^{-2r}\log(\bar k)$, which is $O(\bar k^{-1})$ because $2r=s+1>1$. If $2r-2-s>-1$, the sum is at most $C\bar k^{2r-1-s}$, and the display is at most $C\bar k^{-1-s}\le C\bar k^{-1}$. Thus
\begin{equation}
  \E\left\|\sum_{j=0}^{k-1}\beta_j\Pi^B_{j+1,k}\nu_{j+1}\right\|^2
  \le C\bar k^{-1}. \label{eq:nu-centered-bound}
\end{equation}
Combining the five displayed bounds and using $\|z_1+\cdots+z_5\|^2\le5\sum_i\|z_i\|^2$ proves \eqref{eq:centering-rate}. If $m_k\equiv0$, the term \eqref{eq:m-centered-bound} is absent, and the rate is $C\bar k^{-1}$.
\end{proof}

\section{Proof of online bias tracking}
\label{app:bias-tracking-proof}

This appendix proves Theorem~\ref{thm:bias-tracking}. The argument has two steps. First, the tracker $M_k$ inherits the $L^2$ decay of the nonlinear input $R(X_k)$. Second, the corrected residual $R(X_k)-M_k$ is written as a scaled difference of consecutive tracker values, so its slow convolution is controlled by Abel summation.

For $0\le i\le k$, define the scalar tracking transition
\[
  \Pi^\gamma_{i,k}
  =
  \begin{cases}
    \prod_{\ell=i}^{k-1}(1-\gamma_\ell), & i<k,\\
    1, & i=k.
  \end{cases}
\]

\begin{lemma}[Moment bound for the bias tracker]
\label{lem:bias-tracker-moment}
Under the assumptions of Theorem~\ref{thm:bias-tracking},
\[
  \big(\E\norm{M_k}^2\big)^{1/2}\le C\bar k^{-s/2},
  \qquad k\ge1.
\]
Consequently, $\E\norm{M_k}^2\le C\bar k^{-s}$.
\end{lemma}

\begin{proof}
Unrolling \eqref{eq:bias-tracker} gives
\[
  M_k=\Pi^\gamma_{0,k}M_0+
  \sum_{i=0}^{k-1}\gamma_i\Pi^\gamma_{i+1,k}R(X_i).
\]
By Assumption~\ref{ass:remainder} and Lemma~\ref{lem:fast-moments},
\[
  \big(\E\norm{R(X_i)}^2\big)^{1/2}
  \le L_R\big(\E\norm{X_i}^{2+2\rho}\big)^{1/2}
  \le C\alpha_i^{(1+\rho)/2}
  \le C\bar i^{-s/2}. \label{eq:R-L2-decay-bias}
\]
Thus, by the $L^2$ triangle inequality,
\begin{equation}
  \norm{M_k}_{L^2}
  \le \Pi^\gamma_{0,k}\norm{M_0}
     +C\sum_{i=0}^{k-1}\gamma_i\Pi^\gamma_{i+1,k}\bar i^{-s/2}. \label{eq:M-unroll-bound}
\end{equation}
We bound the two terms on the right.

Because $c<1$ and $\gamma_i=\gamma_0\bar i^{-c}$,
\[
  \sum_{i=0}^{k-1}\gamma_i
  \ge c_0\bar k^{1-c}
\]
for all $k$ and a constant $c_0>0$ after changing constants for small $k$. Hence
\[
  \Pi^\gamma_{0,k}
  \le \exp\!\left(-\sum_{i=0}^{k-1}\gamma_i\right)
  \le \exp(-c_0\bar k^{1-c})
  \le C\bar k^{-s/2}.
\]
This controls the initial term.

For the convolution term, split the sum in \eqref{eq:M-unroll-bound} at $\lfloor k/2\rfloor$. If $i\le k/2$, then
\[
  \sum_{\ell=i+1}^{k-1}\gamma_\ell\ge c_1\bar k^{1-c}
\]
for some $c_1>0$, and therefore
\[
  \gamma_i\Pi^\gamma_{i+1,k}\bar i^{-s/2}
  \le C\bar i^{-c-s/2}e^{-c_1\bar k^{1-c}}.
\]
Summing over $i\le k/2$ gives at most $C\bar k^{-s/2}$, because the stretched-exponential factor dominates every polynomial. If $i>k/2$, then $\bar i^{-s/2}\le C\bar k^{-s/2}$, and
\[
  \sum_{i=\lfloor k/2\rfloor+1}^{k-1}\gamma_i\Pi^\gamma_{i+1,k}
  \le \sum_{i=0}^{k-1}\gamma_i\Pi^\gamma_{i+1,k}
  =1-\Pi^\gamma_{0,k}
  \le 1.
\]
Therefore the recent half of the sum is also at most $C\bar k^{-s/2}$. Combining the initial, early, and recent bounds proves the lemma.
\end{proof}

\begin{lemma}[Abel bound for the tracked residual]
\label{lem:tracked-residual-abel}
Under the assumptions of Theorem~\ref{thm:bias-tracking},
\begin{equation}
  \E\left\|
  \sum_{j=0}^{k-1}\beta_j\Pi^B_{j+1,k}\big(R(X_j)-M_j\big)
  \right\|^2
  \le C\bar k^{-1}. \label{eq:tracked-residual-bound}
\end{equation}
\end{lemma}

\begin{proof}
From \eqref{eq:bias-tracker},
\[
  R(X_j)-M_j=\gamma_j^{-1}(M_{j+1}-M_j).
\]
Define
\[
  q_j=\frac{\beta_j}{\gamma_j}=\frac{\beta_0}{\gamma_0}\bar j^{c-1},
  \qquad
  Q_{j,k}=q_j\Pi^B_{j+1,k},
  \qquad 0\le j\le k-1.
\]
Then
\[
  S_k\coloneqq
  \sum_{j=0}^{k-1}\beta_j\Pi^B_{j+1,k}\big(R(X_j)-M_j\big)
  =\sum_{j=0}^{k-1}Q_{j,k}(M_{j+1}-M_j).
\]
Abel summation gives
\begin{equation}
  S_k=-Q_{0,k}M_0+Q_{k-1,k}M_k+
      \sum_{j=1}^{k-1}(Q_{j-1,k}-Q_{j,k})M_j. \label{eq:tracked-abel}
\end{equation}
We first record coefficient bounds. Lemma~\ref{lem:transition} gives
\begin{equation}
  \norm{Q_{j,k}}
  \le C\bar k^{-r}\bar j^{r+c-1}. \label{eq:Qgamma-size}
\end{equation}
Moreover, the same calculation as in Lemma~\ref{lem:abel-coefficients}, with $a$ replaced by $c$, yields
\begin{equation}
  \norm{Q_{j,k}-Q_{j-1,k}}
  \le C\bar k^{-r}\bar j^{r+c-2},
  \qquad 1\le j\le k-1. \label{eq:Qgamma-var}
\end{equation}
Indeed, $|q_j-q_{j-1}|\le C\bar j^{c-2}$ and $q_{j-1}\beta_j\le C\bar j^{c-2}$.

We now bound the three terms in \eqref{eq:tracked-abel}. The initial boundary satisfies
\[
  \norm{Q_{0,k}M_0}^2\le C\bar k^{-2r}\norm{M_0}^2\le C\bar k^{-1},
\]
because $r>1/2$. For the terminal boundary, \eqref{eq:Qgamma-size} and Lemma~\ref{lem:bias-tracker-moment} give
\[
  \E\norm{Q_{k-1,k}M_k}^2
  \le C\bar k^{2c-2}\E\norm{M_k}^2
  \le C\bar k^{2c-2-s}.
\]
The condition $c<(1+s)/2$ implies $2c-2-s<-1$, and hence this term is at most $C\bar k^{-1}$.

For the interior term, the $L^2$ triangle inequality, \eqref{eq:Qgamma-var}, and Lemma~\ref{lem:bias-tracker-moment} imply
\begin{align*}
  \left\|\sum_{j=1}^{k-1}(Q_{j-1,k}-Q_{j,k})M_j\right\|_{L^2}
  &\le C\bar k^{-r}\sum_{j=1}^{k-1}\bar j^{r+c-2}\bar j^{-s/2}  \\
  &=C\bar k^{-r}\sum_{j=1}^{k-1}\bar j^{r+c-2-s/2}.
\end{align*}
If $r+c-2-s/2<-1$, the sum is bounded, so the squared expression is at most $C\bar k^{-2r}\le C\bar k^{-1}$. If $r+c-2-s/2=-1$, the sum is logarithmic, and $C\bar k^{-2r}(\log\bar k)^2\le C\bar k^{-1}$ because $r>1/2$. If $r+c-2-s/2>-1$, the sum is at most $C\bar k^{r+c-1-s/2}$, and the square is at most
\[
  C\bar k^{-2r}\bar k^{2r+2c-2-s}
  =C\bar k^{2c-2-s}
  \le C\bar k^{-1},
\]
again by $c<(1+s)/2$. Combining the three terms in \eqref{eq:tracked-abel} proves \eqref{eq:tracked-residual-bound}.
\end{proof}

\begin{proof}[Proof of Theorem~\ref{thm:bias-tracking}]
The estimate for $M_k$ is Lemma~\ref{lem:bias-tracker-moment}. It remains to prove the slow-rate estimate. Unrolling \eqref{eq:deb-slow} gives
\begin{align*}
  \widetilde Y_k
  &=\Pi^B_{0,k}\widetilde Y_0
    +\sum_{j=0}^{k-1}\beta_j\Pi^B_{j+1,k}HX_j \\
  &\quad
    +\sum_{j=0}^{k-1}\beta_j\Pi^B_{j+1,k}\big(R(X_j)-M_j\big)
    +\sum_{j=0}^{k-1}\beta_j\Pi^B_{j+1,k}\zeta_{j+1}.
\end{align*}
The initial term satisfies
\[
  \norm{\Pi^B_{0,k}\widetilde Y_0}^2\le C\bar k^{-2r}\le C\bar k^{-1}
\]
because $r>1/2$. Lemma~\ref{lem:linear-leakage} bounds the locally linear leakage by $C\bar k^{-1}$ in mean square. Lemma~\ref{lem:tracked-residual-abel} bounds the corrected nonlinear residual by $C\bar k^{-1}$ in mean square. Finally, Lemma~\ref{lem:nonlinear-slow-noise} bounds the slow martingale noise by $C\bar k^{-1}$ in mean square. Using $\norm{z_1+z_2+z_3+z_4}^2\le4\sum_i\norm{z_i}^2$ proves
\[
  \E\norm{\widetilde Y_k}^2\le C\bar k^{-1}.
\]
This proves the theorem.
\end{proof}

\section{Proof of the localized nonlinear transfer theorem}
\label{app:localized-transfer}

This appendix proves Theorem~\ref{thm:localized-transfer} and Propositions~\ref{prop:verification} and~\ref{prop:generic-verification}. We first record a weighted-sum estimate that allows borderline logarithms.

\begin{lemma}[General weighted convolution]\label{lem:general-weighted}
Let $u\ge0$ and $r>0$. Then, for all $k\ge1$,
\begin{equation}
  \sum_{j=0}^{k-1}\beta_j\norm{\Pi^B_{j+1,k}}\bar j^{-u}
  \le
  C\begin{cases}
    \bar k^{-u}, & u<r,\\
    \bar k^{-r}\log(\bar k), & u=r,\\
    \bar k^{-r}, & u>r.
  \end{cases} \label{eq:general-weighted}
\end{equation}
In particular, if $\min\{u,r\}>v$, then the square of the left-hand side is at most $C\bar k^{-2v}$.
\end{lemma}

\begin{proof}
By Lemma~\ref{lem:transition} and the inequality $\overline{j+1}\le2\bar j$,
\[
  \sum_{j=0}^{k-1}\beta_j\norm{\Pi^B_{j+1,k}}\bar j^{-u}
  \le C\bar k^{-r}\sum_{j=0}^{k-1}\bar j^{r-u-1}.
\]
If $r-u>0$, the sum is at most $C\bar k^{r-u}$. If $r-u=0$, the sum is at most $C\log(\bar k)$. If $r-u<0$, the sum is bounded by a constant because $k_0\ge2$ and the exponent $r-u-1$ is strictly less than $-1$. These three cases give \eqref{eq:general-weighted}. If $\min\{u,r\}>v$, the first and third cases give $C\bar k^{-2v}$ after squaring, while the borderline case gives $C\bar k^{-2r}(\log \bar k)^2\le C\bar k^{-2v}$ because $r>v$.
\end{proof}

\begin{lemma}[Perturbed linear leakage]\label{lem:perturbed-leakage}
Under Assumption~\ref{ass:nonlinear-reduction},
\begin{equation}
  \E\left\|\sum_{j=0}^{k-1}\beta_j\Pi^B_{j+1,k}HX_j\right\|^2
  \le C\left(\bar k^{-1}+\bar k^{-s}\right). \label{eq:perturbed-leakage-bound}
\end{equation}
\end{lemma}

\begin{proof}
Let $C_H=HA^{-1}$. Rearranging the perturbed fast recursion \eqref{eq:perturbed-fast} gives
\begin{equation}
  AX_j
  =\alpha_j^{-1}(X_j-X_{j+1})+\xi_{j+1}+p_{j+1}
   +\frac{\beta_j}{\alpha_j}C_\zeta\zeta_{j+1}
   +\frac{\beta_j}{\alpha_j}d_j. \label{eq:perturbed-poisson}
\end{equation}
Multiplying by $HA^{-1}=C_H$ and summing against the slow weights yields
\begin{align}
  \sum_{j=0}^{k-1}\beta_j\Pi^B_{j+1,k}HX_j
  &=T_k+M^\xi_k+P_k+M^\zeta_k+D_k, \label{eq:perturbed-decomposition}\\
  T_k&\coloneqq\sum_{j=0}^{k-1}Q_{j,k}(X_j-X_{j+1}), \label{eq:transfer-T}\\
  M^\xi_k&\coloneqq\sum_{j=0}^{k-1}\beta_j\Pi^B_{j+1,k}C_H\xi_{j+1}, \\
  P_k&\coloneqq\sum_{j=0}^{k-1}\beta_j\Pi^B_{j+1,k}C_Hp_{j+1}, \\
  M^\zeta_k&\coloneqq\sum_{j=0}^{k-1}\beta_j\Pi^B_{j+1,k}C_H\frac{\beta_j}{\alpha_j}C_\zeta\zeta_{j+1}, \\
  D_k&\coloneqq\sum_{j=0}^{k-1}\beta_j\Pi^B_{j+1,k}C_H\frac{\beta_j}{\alpha_j}d_j,
\end{align}
where $Q_{j,k}$ is defined in \eqref{eq:Q-def}.

The term $T_k$ is the same Abel boundary term as in Lemma~\ref{lem:linear-leakage}. Its proof used only the coefficient bounds in Lemma~\ref{lem:abel-coefficients} and the moment estimate $\E\norm{X_j}^2\le C\alpha_j$. Therefore
\begin{equation}
  \E\norm{T_k}^2\le C\bar k^{-1}. \label{eq:T-transfer-bound}
\end{equation}
The martingale term $M^\xi_k$ is also the same as before, hence
\begin{equation}
  \E\norm{M^\xi_k}^2\le C\bar k^{-1}. \label{eq:Mxi-transfer-bound}
\end{equation}

For $P_k$, the triangle inequality in $L^2$ and \eqref{eq:p-bound} give
\[
  \big(\E\norm{P_k}^2\big)^{1/2}
  \le C\sum_{j=0}^{k-1}\beta_j\norm{\Pi^B_{j+1,k}}\bar j^{-\mu_p}.
\]
By Lemma~\ref{lem:general-weighted}, this is at most $C\bar k^{-\nu_p}$ up to a logarithm, where $\nu_p=\min\{\mu_p,r\}$. Since $\mu_p>\delta/2$ and $r>1/2\ge\delta/2$, we have $\nu_p>\delta/2$. Therefore
\begin{equation}
  \E\norm{P_k}^2\le C\bar k^{-\delta}
  \le C\left(\bar k^{-1}+\bar k^{-s}\right). \label{eq:P-transfer-bound}
\end{equation}

For $M^\zeta_k$, martingale cross terms vanish. Since $\beta_j/\alpha_j=C\bar j^{a-1}$,
\begin{align*}
  \E\norm{M^\zeta_k}^2
  &\le C\sum_{j=0}^{k-1}\beta_j^2\norm{\Pi^B_{j+1,k}}^2\left(\frac{\beta_j}{\alpha_j}\right)^2 \\
  &\le C\bar k^{-2r}\sum_{j=0}^{k-1}\bar j^{2r+2a-4}.
\end{align*}
If $2r+2a-4<-1$, the last sum is bounded by a constant and the result is at most $C\bar k^{-2r}\le C\bar k^{-1}$. If $2r+2a-4=-1$, the sum is at most $C\log(\bar k)$ and the result is at most $C\bar k^{-2r}\log(\bar k)\le C\bar k^{-1}$ because $r>1/2$. If $2r+2a-4>-1$, the sum is at most $C\bar k^{2r+2a-3}$, so the result is at most $C\bar k^{2a-3}$. Since $a<1$, $3-2a>1$, and again the term is $O(\bar k^{-1})$. Thus
\begin{equation}
  \E\norm{M^\zeta_k}^2\le C\bar k^{-1}. \label{eq:Mzeta-transfer-bound}
\end{equation}

It remains to bound $D_k$. By \eqref{eq:d-bound},
\begin{align*}
  \big(\E\norm{D_k}^2\big)^{1/2}
  &\le C\sum_{j=0}^{k-1}\beta_j\norm{\Pi^B_{j+1,k}}\bar j^{a-1}
      \big(\bar j^{-a/2}+\bar j^{-\delta/2}\big) \\
  &= C\sum_{j=0}^{k-1}\beta_j\norm{\Pi^B_{j+1,k}}\bar j^{-(1-a/2)}
    +C\sum_{j=0}^{k-1}\beta_j\norm{\Pi^B_{j+1,k}}\bar j^{-(1-a+\delta/2)}.
\end{align*}
The first exponent is $u_1=1-a/2$, and $2u_1=2-a\ge1\ge\delta$. The second exponent is $u_2=1-a+\delta/2$, and $u_2>\delta/2$ because $a<1$. Since $r>1/2\ge\delta/2$, Lemma~\ref{lem:general-weighted} implies that both sums have squares bounded by $C\bar k^{-\delta}$. Hence
\begin{equation}
  \E\norm{D_k}^2\le C\bar k^{-\delta}
  \le C\left(\bar k^{-1}+\bar k^{-s}\right). \label{eq:D-transfer-bound}
\end{equation}

Combining \eqref{eq:perturbed-decomposition}--\eqref{eq:D-transfer-bound} and using the elementary inequality $\|z_1+\cdots+z_5\|^2\le5\sum_{i=1}^5\|z_i\|^2$ proves \eqref{eq:perturbed-leakage-bound}.
\end{proof}

\begin{proof}[Proof of Theorem~\ref{thm:localized-transfer}]
Unroll \eqref{eq:perturbed-slow}:
\begin{align*}
  Y_k
  &=\Pi^B_{0,k}Y_0
    +\sum_{j=0}^{k-1}\beta_j\Pi^B_{j+1,k}HX_j
    +\sum_{j=0}^{k-1}\beta_j\Pi^B_{j+1,k}R(X_j)\\
  &\qquad
    +\sum_{j=0}^{k-1}\beta_j\Pi^B_{j+1,k}\zeta_{j+1}
    +\sum_{j=0}^{k-1}\beta_j\Pi^B_{j+1,k}e_j.
\end{align*}
The initial term satisfies $\norm{\Pi^B_{0,k}Y_0}^2\le C\bar k^{-2r}\le C\bar k^{-1}$ because $r>1/2$. The locally linear leakage term is bounded by Lemma~\ref{lem:perturbed-leakage}. The nonlinear remainder and slow martingale noise are bounded exactly as in Lemma~\ref{lem:nonlinear-slow-noise}, using \eqref{eq:transfer-fast-moments} in place of Lemma~\ref{lem:fast-moments}. Thus they contribute $C(\bar k^{-1}+\bar k^{-s})$ and $C\bar k^{-1}$, respectively.

For the additional slow perturbation, the $L^2$ triangle inequality and \eqref{eq:e-bound} give
\[
  \left(\E\left\|\sum_{j=0}^{k-1}\beta_j\Pi^B_{j+1,k}e_j\right\|^2\right)^{1/2}
  \le C\sum_{j=0}^{k-1}\beta_j\norm{\Pi^B_{j+1,k}}\bar j^{-\delta/2}.
\]
Since $r>1/2\ge\delta/2$, Lemma~\ref{lem:general-weighted} implies that the square of this term is at most $C\bar k^{-\delta}$. Because $\bar k^{-\delta}\le C(\bar k^{-1}+\bar k^{-s})$, all terms are bounded by the right-hand side of \eqref{eq:localized-transfer-rate}. This proves the theorem.
\end{proof}

\begin{proof}[Proof of Proposition~\ref{prop:verification}]
We verify Assumption~\ref{ass:nonlinear-reduction} under the displayed local error recursions. Let
\[
  X_k=x_k-\lambda(y_k),
  \qquad
  Y_k=y_k-y^*,
  \qquad
  \xi_{k+1}=M^x_{k+1},
  \qquad
  \zeta_{k+1}=M^y_{k+1}.
\]
The slow recursion follows immediately from \eqref{eq:g-local}:
\[
  Y_{k+1}
  =Y_k+\beta_k\big(-BY_k+HX_k+R(X_k)+E(X_k,y_k)+\zeta_{k+1}\big).
\]
Thus \eqref{eq:perturbed-slow} holds with $e_k=E(X_k,y_k)$. By \eqref{eq:g-local}, \eqref{eq:baseline-localization}, and Lyapunov's inequality,
\[
  \big(\E\norm{e_k}^2\big)^{1/2}
  \le L_g\big(\E\norm{X_k}^{2+2\rho}\big)^{1/2}
  \le L_g\big(\E\norm{X_k}^4\big)^{(1+\rho)/4}
  \le C\alpha_k^{(1+\rho)/2}
  =C\bar k^{-s/2}
  \le C\bar k^{-\delta/2}.
\]
This is \eqref{eq:e-bound}. The first estimate in \eqref{eq:transfer-fast-moments} follows from \eqref{eq:baseline-localization} by Jensen's inequality:
\[
  \E\norm{X_k}^2\le (\E\norm{X_k}^4)^{1/2}\le C\alpha_k.
\]
The second estimate in \eqref{eq:transfer-fast-moments} is the same interpolation calculation:
\[
  \E\norm{X_k}^{2+2\rho}
  \le (\E\norm{X_k}^4)^{(1+\rho)/2}
  \le C\alpha_k^{1+\rho}.
\]

It remains to derive the perturbed fast recursion and the bounds on $p_{k+1}$ and $d_k$. By Taylor's theorem applied to $\lambda$ between $y_k$ and $y_{k+1}$,
\begin{equation}
  \lambda(y_{k+1})
  =\lambda(y_k)+D\lambda(y_k)(y_{k+1}-y_k)+r^\lambda_{k+1},
  \qquad
  \norm{r^\lambda_{k+1}}\le C\norm{y_{k+1}-y_k}^2, \label{eq:lambda-taylor}
\end{equation}
where the constant $C$ is the local bound on the Hessian of $\lambda$. Since
\[
  y_{k+1}-y_k=\beta_k\big(g(x_k,y_k)+\zeta_{k+1}\big),
\]
and since \eqref{eq:h-local} gives $h(x_k,y_k)=-AX_k+P(X_k,y_k)$, we obtain
\begin{align*}
  X_{k+1}
  &=x_{k+1}-\lambda(y_{k+1})\\
  &=x_k+\alpha_k\big(h(x_k,y_k)+\xi_{k+1}\big)
    -\lambda(y_k)-D\lambda(y_k)\beta_k\big(g(x_k,y_k)+\zeta_{k+1}\big)-r^\lambda_{k+1}\\
  &=(I-\alpha_kA)X_k+\alpha_k\xi_{k+1}+\alpha_kP(X_k,y_k)
    -\beta_kD\lambda(y^*)\zeta_{k+1}+\beta_k d_k,
\end{align*}
where
\begin{equation}
  d_k
  =-D\lambda(y_k)g(x_k,y_k)
   -\big(D\lambda(y_k)-D\lambda(y^*)\big)\zeta_{k+1}
   -\beta_k^{-1}r^\lambda_{k+1}. \label{eq:d-verification-def}
\end{equation}
Thus \eqref{eq:perturbed-fast} holds with
\[
  p_{k+1}=P(X_k,y_k),
  \qquad
  C_\zeta=-D\lambda(y^*).
\]
The bound on $p_{k+1}$ follows from \eqref{eq:h-local} and \eqref{eq:baseline-localization}:
\[
  \big(\E\norm{p_{k+1}}^2\big)^{1/2}
  \le L_h\big(\E\norm{X_k}^4\big)^{1/2}
  \le C\alpha_k
  =C\bar k^{-a}.
\]
Because $\delta\le1$ and $a>1/2$, we have $a>\delta/2$. Hence \eqref{eq:p-bound} holds with $\mu_p=a$.

We next bound $d_k$. First, \eqref{eq:g-local}, Assumption~\ref{ass:remainder}, the already-proved fast moment bounds, and the baseline estimate for $Y_k$ imply
\begin{align*}
  \big(\E\norm{g(x_k,y_k)}^2\big)^{1/2}
  &\le C\big(\E\norm{Y_k}^2\big)^{1/2}
       +C\big(\E\norm{X_k}^2\big)^{1/2}
       +C\big(\E\norm{X_k}^{2+2\rho}\big)^{1/2}\\
  &\le C\bar k^{-a/2}+C\bar k^{-s/2}
  \le C\big(\bar k^{-a/2}+\bar k^{-\delta/2}\big). \label{eq:g-l2-verification}
\end{align*}
Since $D\lambda$ is locally bounded, the first term in \eqref{eq:d-verification-def} satisfies the same bound.

For the second term in \eqref{eq:d-verification-def}, the local Lipschitz continuity of $D\lambda$ gives
\[
  \norm{D\lambda(y_k)-D\lambda(y^*)}\le C\norm{Y_k}.
\]
Using the conditional second-moment bound on $\zeta_{k+1}$,
\[
  \E\norm{(D\lambda(y_k)-D\lambda(y^*))\zeta_{k+1}}^2
  \le C\E\big[\norm{Y_k}^2\E(\norm{\zeta_{k+1}}^2\mid\mathcal F_k)\big]
  \le C\E\norm{Y_k}^2
  \le C\bar k^{-a}.
\]
Thus this term has $L^2$ norm at most $C\bar k^{-a/2}$.

For the Taylor remainder term, \eqref{eq:lambda-taylor} gives
\[
  \beta_k^{-1}\norm{r^\lambda_{k+1}}
  \le C\beta_k\norm{g(x_k,y_k)+\zeta_{k+1}}^2.
\]
By the localization hypothesis, the process stays in the neighborhood where \eqref{eq:g-local} holds. Hence $g$ is locally bounded there, and $\zeta_{k+1}$ has a uniformly bounded fourth conditional moment. Therefore
\[
  \big(\E\norm{\beta_k^{-1}r^\lambda_{k+1}}^2\big)^{1/2}
  \le C\beta_k
  \le C\bar k^{-1}
  \le C\bar k^{-a/2},
\]
because $a<1$. Combining the three estimates proves \eqref{eq:d-bound}. Therefore Assumption~\ref{ass:nonlinear-reduction} holds, and Theorem~\ref{thm:localized-transfer} gives the claimed rate.
\end{proof}

\begin{proof}[Proof of Proposition~\ref{prop:generic-verification}]
The proof follows the same coordinate calculation as Proposition~\ref{prop:verification}, but we spell out the estimates because the mixed terms are the point of the result. Define
\[
  X_k=x_k-\lambda(y_k),
  \qquad
  Y_k=y_k-y^*,
  \qquad
  \xi_{k+1}=M^x_{k+1},
  \qquad
  \zeta_{k+1}=M^y_{k+1}.
\]
By Jensen's inequality and \eqref{eq:baseline-generic},
\begin{equation}
  \E\norm{X_k}^2\le C\alpha_k,
  \qquad
  \E\norm{Y_k}^2\le C\bar k^{-a}. \label{eq:generic-second-moments}
\end{equation}
By Lyapunov's inequality,
\begin{equation}
  \E\norm{X_k}^{2+2\rho}
  \le \big(\E\norm{X_k}^4\big)^{(1+\rho)/2}
  \le C\alpha_k^{1+\rho}. \label{eq:generic-x-interp}
\end{equation}
Thus the fast moment requirements \eqref{eq:transfer-fast-moments} hold.

The slow equation follows directly from \eqref{eq:g-generic}:
\[
  Y_{k+1}=(I-\beta_kB)Y_k+\beta_k\big(HX_k+R(X_k)+\zeta_{k+1}+e_k\big),
  \qquad
  e_k=E(X_k,y_k).
\]
We bound $e_k$ in $L^2$. By \eqref{eq:g-generic}, Cauchy's inequality, \eqref{eq:baseline-generic}, and \eqref{eq:generic-x-interp},
\begin{align*}
  \big(\E\norm{e_k}^2\big)^{1/2}
  &\le C\big(\E\norm{X_k}^{2+2\rho}\big)^{1/2}
       +C\big(\E\norm{X_k}^2\norm{Y_k}^2\big)^{1/2}
       +C\big(\E\norm{Y_k}^4\big)^{1/2}\\
  &\le C\bar k^{-s/2}
       +C\big(\E\norm{X_k}^4\big)^{1/4}\big(\E\norm{Y_k}^4\big)^{1/4}
       +C\bar k^{-a}\\
  &\le C\bar k^{-s/2}+C\bar k^{-a}+C\bar k^{-a}.
\end{align*}
Let $\delta=\min\{1,s\}$. Since $a>1/2\ge \delta/2$ and $s/2\ge \delta/2$, the last line is at most $C\bar k^{-\delta/2}$. Hence \eqref{eq:e-bound} holds.

We next derive the fast equation. Taylor's theorem for $\lambda$ gives
\begin{equation}
  \lambda(y_{k+1})
  =\lambda(y_k)+D\lambda(y_k)(y_{k+1}-y_k)+r^\lambda_{k+1},
  \qquad
  \norm{r^\lambda_{k+1}}\le C\norm{y_{k+1}-y_k}^2. \label{eq:lambda-taylor-generic}
\end{equation}
Using \eqref{eq:h-generic} and $y_{k+1}-y_k=\beta_k(g(x_k,y_k)+\zeta_{k+1})$, we get
\begin{align*}
  X_{k+1}
  &=x_{k+1}-\lambda(y_{k+1})\\
  &=x_k+\alpha_k\big(h(x_k,y_k)+\xi_{k+1}\big)-\lambda(y_k)
    -D\lambda(y_k)\beta_k\big(g(x_k,y_k)+\zeta_{k+1}\big)-r^\lambda_{k+1}\\
  &=(I-\alpha_kA)X_k+\alpha_k\xi_{k+1}+\alpha_kP(X_k,y_k)
    -\beta_kD\lambda(y^*)\zeta_{k+1}+\beta_kd_k,
\end{align*}
where
\begin{equation}
  d_k=-D\lambda(y_k)g(x_k,y_k)
      -\big(D\lambda(y_k)-D\lambda(y^*)\big)\zeta_{k+1}
      -\beta_k^{-1}r^\lambda_{k+1}. \label{eq:d-generic-def}
\end{equation}
Thus \eqref{eq:perturbed-fast} holds with
\[
  p_{k+1}=P(X_k,y_k),
  \qquad
  C_\zeta=-D\lambda(y^*).
\]
For $p_{k+1}$, \eqref{eq:h-generic}, Cauchy's inequality, and \eqref{eq:baseline-generic} imply
\begin{align*}
  \big(\E\norm{p_{k+1}}^2\big)^{1/2}
  &\le C\big(\E\norm{X_k}^4\big)^{1/2}
      +C\big(\E\norm{X_k}^2\norm{Y_k}^2\big)^{1/2}\\
  &\le C\alpha_k+C\big(\E\norm{X_k}^4\big)^{1/4}\big(\E\norm{Y_k}^4\big)^{1/4}
   \le C\bar k^{-a}.
\end{align*}
Because $a>1/2\ge\delta/2$, \eqref{eq:p-bound} holds with $\mu_p=a$.

It remains to prove \eqref{eq:d-bound}. First, by \eqref{eq:g-generic}, Assumption~\ref{ass:remainder}, \eqref{eq:generic-second-moments}, \eqref{eq:generic-x-interp}, and the bound just proved for $e_k$,
\begin{align*}
  \big(\E\norm{g(x_k,y_k)}^2\big)^{1/2}
  &\le C\big(\E\norm{Y_k}^2\big)^{1/2}
       +C\big(\E\norm{X_k}^2\big)^{1/2}
       +C\big(\E\norm{X_k}^{2+2\rho}\big)^{1/2}
       +C\big(\E\norm{e_k}^2\big)^{1/2}\\
  &\le C\bar k^{-a/2}+C\bar k^{-\delta/2}.
\end{align*}
The derivative $D\lambda$ is locally bounded, so the first term in \eqref{eq:d-generic-def} has this same $L^2$ bound.

For the second term in \eqref{eq:d-generic-def}, boundedness of the Hessian of $\lambda$ gives
\[
  \norm{D\lambda(y_k)-D\lambda(y^*)}\le C\norm{Y_k}.
\]
Using the conditional second-moment bound on $\zeta_{k+1}$,
\begin{align*}
  \E\norm{(D\lambda(y_k)-D\lambda(y^*))\zeta_{k+1}}^2
  &\le C\E\big[\norm{Y_k}^2\E(\norm{\zeta_{k+1}}^2\mid\mathcal F_k)\big]
   \le C\E\norm{Y_k}^2
   \le C\bar k^{-a}.
\end{align*}
Therefore this term has $L^2$ norm at most $C\bar k^{-a/2}$.

For the Taylor remainder term, \eqref{eq:lambda-taylor-generic} gives
\[
  \beta_k^{-1}\norm{r^\lambda_{k+1}}
  \le C\beta_k\norm{g(x_k,y_k)+\zeta_{k+1}}^2.
\]
The localized process remains in the neighborhood on which $g$ is bounded, and $\zeta_{k+1}$ has a uniformly bounded fourth conditional moment. Hence
\[
  \big(\E\norm{\beta_k^{-1}r^\lambda_{k+1}}^2\big)^{1/2}
  \le C\beta_k
  \le C\bar k^{-1}
  \le C\bar k^{-a/2},
\]
because $a<1$. Combining the three bounds for the terms in \eqref{eq:d-generic-def} yields
\[
  \big(\E\norm{d_k}^2\big)^{1/2}
  \le C\big(\bar k^{-a/2}+\bar k^{-\delta/2}\big),
\]
which is \eqref{eq:d-bound}. All parts of Assumption~\ref{ass:nonlinear-reduction} have now been verified, and Theorem~\ref{thm:localized-transfer} gives the claimed rate.
\end{proof}

\section{Proof of the lower bound}\label{app:lower}

We prove Theorem \ref{thm:lower}. The proof uses elementary properties of the scalar Gaussian fast recursion.

\begin{lemma}[Fast variance lower bound]\label{lem:variance-lower}
Consider the scalar recursion
\[
  X_{k+1}=(1-\lambda\alpha_k)X_k+\alpha_k\xi_{k+1},
  \qquad X_0=0,
\]
where $\xi_k$ are i.i.d. standard Gaussian random variables and $0<\lambda\alpha_k\le1/2$. Then there are constants $c_v,C_v>0$ and $k_v\ge1$ such that, for all $k\ge k_v$,
\begin{equation}
  c_v\alpha_k\le \E[X_k^2]\le C_v\alpha_k. \label{eq:variance-two-sided}
\end{equation}
Consequently, for every $p>0$ there exists $c_p>0$ such that, for all $k\ge k_v$,
\begin{equation}
  \E|X_k|^p\ge c_p\alpha_k^{p/2}. \label{eq:gaussian-p-lower}
\end{equation}
\end{lemma}

\begin{proof}
Because the recursion is linear and the driving noises are Gaussian, $X_k$ is a centered Gaussian random variable. Let $v_k=\E[X_k^2]$. Then
\begin{equation}
  v_{k+1}=(1-\lambda\alpha_k)^2v_k+\alpha_k^2. \label{eq:v-recursion}
\end{equation}
The upper bound $v_k\le C_v\alpha_k$ follows from the same second-moment induction used in Lemma \ref{lem:fast-moments}; we repeat the one-line argument. Since $(1-\lambda\alpha_k)^2\le1-\lambda\alpha_k$ and $\alpha_{k+1}\ge\alpha_k(1-(\lambda/2)\alpha_k)$ for all large $k$, choosing $C_v\ge2/\lambda$ yields
\[
  v_k\le C_v\alpha_k
  \quad\Rightarrow\quad
  v_{k+1}\le C_v\alpha_k(1-\lambda\alpha_k)+\alpha_k^2
  \le C_v\alpha_{k+1}.
\]
Increasing $C_v$ handles finitely many initial indices.

For the lower bound, unroll \eqref{eq:v-recursion}:
\begin{equation}
  v_k=\sum_{i=0}^{k-1}\alpha_i^2\prod_{\ell=i+1}^{k-1}(1-\lambda\alpha_\ell)^2. \label{eq:v-unrolled}
\end{equation}
Fix a small constant $\delta\in(0,1)$ to be chosen below and define
\[
  m_k\coloneqq\left\lfloor \frac{\delta}{\alpha_k}\right\rfloor.
\]
Because $a<1$, $m_k=o(k)$, so for all large $k$ we have $m_k\le k/2$. For $i\in\{k-m_k,\ldots,k-1\}$ and all large $k$, monotonicity of $\alpha_j$ and $m_k=o(k)$ imply
\begin{equation}
  \frac12\alpha_k\le \alpha_i\le 2\alpha_k. \label{eq:alpha-window}
\end{equation}
Also, since $\lambda\alpha_\ell\le1/2$, the inequality $\log(1-u)\ge-2u$ for $u\in[0,1/2]$ gives
\[
\begin{aligned}
  \prod_{\ell=i+1}^{k-1}(1-\lambda\alpha_\ell)^2
  &=\exp\left(2\sum_{\ell=i+1}^{k-1}\log(1-\lambda\alpha_\ell)\right)\\
  &\ge \exp\left(-4\lambda\sum_{\ell=i+1}^{k-1}\alpha_\ell\right).
\end{aligned}
\]
For $i\ge k-m_k$, \eqref{eq:alpha-window} gives
\[
  \sum_{\ell=i+1}^{k-1}\alpha_\ell\le 2m_k\alpha_k\le 2\delta.
\]
Therefore every product in the last $m_k$ terms is at least $e^{-8\lambda\delta}$. Restricting the sum \eqref{eq:v-unrolled} to these terms,
\[
  v_k\ge \sum_{i=k-m_k}^{k-1}\alpha_i^2 e^{-8\lambda\delta}
  \ge m_k\left(\frac12\alpha_k\right)^2e^{-8\lambda\delta}.
\]
For all large $k$, $m_k\ge \delta/(2\alpha_k)$, so
\[
  v_k\ge \frac{\delta e^{-8\lambda\delta}}{8}\alpha_k.
\]
This proves the lower variance bound.

Finally, if $X_k\sim N(0,v_k)$, then for any $p>0$,
\[
  \E|X_k|^p = v_k^{p/2}\E|Z|^p,
  \qquad Z\sim N(0,1).
\]
Using $v_k\ge c_v\alpha_k$ proves \eqref{eq:gaussian-p-lower}.
\end{proof}

\begin{proof}[Proof of Theorem \ref{thm:lower}]
Let
\[
  \phi_{i,k}\coloneqq \prod_{\ell=i}^{k-1}(1-b\beta_\ell),\qquad 0\le i\le k,
\]
with the empty product $\phi_{k,k}=1$. Unrolling \eqref{eq:lower-slow} gives
\begin{equation}
  Y_k=\gamma\sum_{j=0}^{k-1}\beta_j\phi_{j+1,k}|X_j|^{1+\rho}
      +\tau\sum_{j=0}^{k-1}\beta_j\phi_{j+1,k}\eta_{j+1}. \label{eq:Y-lower-unrolled}
\end{equation}
Denote the two sums by $P_k$ and $N_k$, respectively, so that $Y_k=\gamma P_k+\tau N_k$.

The process $P_k$ is a measurable function of $\xi_1,\ldots,\xi_k$, whereas $N_k$ is a linear function of $\eta_1,\ldots,\eta_k$. The two noise sequences are independent, so $P_k$ and $N_k$ are independent. Moreover, $\E N_k=0$. Hence
\begin{equation}
  \E[Y_k^2]
  =\gamma^2\E[P_k^2]+\tau^2\E[N_k^2]
  \ge \gamma^2(\E P_k)^2+\tau^2\E[N_k^2]. \label{eq:lower-split}
\end{equation}
We lower bound the two terms separately.

\paragraph{Nonlinear bias term.}
For $j\in\{\lfloor k/2\rfloor,\ldots,k-1\}$ and all large $k$, the product $\phi_{j+1,k}$ is bounded below by a positive constant. To see this, use $\log(1-u)\ge-2u$ for $u\in[0,1/2]$:
\[
\begin{aligned}
  \phi_{j+1,k}
  &=\exp\left(\sum_{\ell=j+1}^{k-1}\log(1-b\beta_\ell)\right)\\
  &\ge \exp\left(-2b\sum_{\ell=j+1}^{k-1}\beta_\ell\right)
   \ge \exp\left(-2b\beta_0\sum_{\ell=\lfloor k/2\rfloor}^{k-1}\frac1{\ell+k_0}\right)\\
  &\ge \exp(-C_b)\eqqcolon c_\phi>0.
\end{aligned}
\]
Also, for these $j$, $\beta_j\ge c/k$ and $\alpha_j\ge c\alpha_k$. Lemma \ref{lem:variance-lower} with $p=1+\rho$ gives
\[
  \E|X_j|^{1+\rho}\ge c\alpha_j^{(1+\rho)/2}
  \ge c\alpha_k^{(1+\rho)/2}.
\]
Therefore
\[
\begin{aligned}
  \E P_k
  &=\sum_{j=0}^{k-1}\beta_j\phi_{j+1,k}\E|X_j|^{1+\rho}\\
  &\ge \sum_{j=\lfloor k/2\rfloor}^{k-1}\beta_j\phi_{j+1,k}\E|X_j|^{1+\rho}\\
  &\ge c\sum_{j=\lfloor k/2\rfloor}^{k-1}\frac1{k}\alpha_k^{(1+\rho)/2}
   \ge c\alpha_k^{(1+\rho)/2}.
\end{aligned}
\]
Since $\alpha_k=\alpha_0\bar k^{-a}$, this implies
\begin{equation}
  (\E P_k)^2\ge c\bar k^{-a(1+\rho)}. \label{eq:P-lower}
\end{equation}

\paragraph{Slow-noise term.}
Because the $\eta_j$ are independent standard Gaussians,
\[
  \E[N_k^2]=\sum_{j=0}^{k-1}\beta_j^2\phi_{j+1,k}^2.
\]
Restricting again to $j\in\{\lfloor k/2\rfloor,\ldots,k-1\}$ and using $\phi_{j+1,k}\ge c_\phi$ and $\beta_j\ge c/k$,
\begin{equation}
  \E[N_k^2]
  \ge \sum_{j=\lfloor k/2\rfloor}^{k-1}\frac{c}{k^2}
  \ge c k^{-1}
  \ge c\bar k^{-1}. \label{eq:N-lower}
\end{equation}

Substituting \eqref{eq:P-lower} and \eqref{eq:N-lower} into \eqref{eq:lower-split} proves
\[
  \E[Y_k^2]\ge c\gamma^2\bar k^{-a(1+\rho)}+c\tau^2\bar k^{-1}
\]
for all sufficiently large $k$. This is \eqref{eq:lower-bound}.
\end{proof}

\section{Additional comments on the threshold}\label{app:threshold}

The upper and lower bounds together show that the exponent
\[
  \min\{1,a(1+\rho)\}
\]
is sharp for the normal-form class. This section records two simple consequences.

\begin{proposition}[Merely Lipschitz remainders cannot decouple]\label{prop:lipschitz}
Set $\rho=0$. In the scalar lower-bound instance with $\gamma>0$, there exists $c>0$ such that
\[
  \E[Y_k^2]\ge c\bar k^{-a}
\]
for all large $k$. Hence no theorem that assumes only a Lipschitz nonlinear remainder can guarantee the decoupled $k^{-1}$ rate uniformly over the normal-form class when $a<1$.
\end{proposition}

\begin{proof}
This is Theorem \ref{thm:lower} with $\rho=0$ and with the nonnegative $\tau^2\bar k^{-1}$ term discarded.
\end{proof}

\begin{proposition}[Second-order remainders decouple for $a>1/2$]\label{prop:second-order}
Set $\rho=1$. Under Assumptions \ref{ass:stability}--\ref{ass:strength}, Theorem \ref{thm:upper} gives
\[
  \E\norm{Y_k}^2\le C\bar k^{-1}.
\]
\end{proposition}

\begin{proof}
When $\rho=1$, $s=2a$. Since $a>1/2$, $s>1$. Thus $\bar k^{-s}\le \bar k^{-1}$, and Theorem \ref{thm:upper} gives the result.
\end{proof}

\end{document}